\documentclass[aps,prd,reprint,nofootinbib,superscriptaddress,showpacs,showkeys]{revtex4-1}
\usepackage{amsfonts}
\usepackage{amsmath,amssymb}
\usepackage{latexsym}
\usepackage[english]{babel}
\usepackage{tikz}
\usepackage{amsthm}

\makeatletter
\g@addto@macro\bfseries{\boldmath}
\makeatother

\numberwithin{equation}{section}
\newtheorem{prop}{Proposition}

\begin{document}
\bibliographystyle{apsrev4-1}
\title{Analytic $SU(N)$ Skyrmions at finite Baryon density}

\author{Pedro D. Alvarez$^\dagger$}
\affiliation{Departamento de Fisica, Universidad de Antofagasta, Aptdo 02800, Chile.}

\author{Sergio L. Cacciatori$^{\dagger\dagger}$}
\affiliation{\textit{Dipartimento di Scienza e Alta Tecnologia, Universit\`a dell'Insubria, via Valleggio 11, Como, Italia Milan, Italia.}}
\affiliation{INFN sezione di Milano, via Celoria 16, 20133 Milan, Italy}

\author{Fabrizio Canfora$^\star$}
\affiliation{\textit{Centro de Estudios Cient\'{\i}ficos (CECS), Casilla 1469, Valdivia, Chile.}}

\author{Bianca L. Cerchiai$^{\star\star}$}
\affiliation{Museo Storico della Fisica e Centro Studi e Ricerche \textquotedblleft Enrico Fermi\textquotedblright , Piazza del Viminale 1,00184 Roma, Italy}
\affiliation{Politecnico di Torino, Dip. DISAT. Corso Duca degli Abruzzi 24, 10129 Torino, Italy}
\affiliation{Istituto Nazionale di Fisica Nucleare (INFN) Sezione diTorino, Italy}
\affiliation{Arnold-Regge Center, via P. Giuria 1, 10125 Torino, Italy}

\email{\small emails: $\star$ canfora@cecs.cl, \\
$\dagger$ pd.alvarez.n@gmail.com,\\
$\dagger\dagger$ sergio.cacciatori@uninsubria.it,\\
$\star\star$ bianca.cerchiai@polito.it}

\begin{abstract}
We construct analytic (3+1)-dimensional Skyr\-mions living at finite Baryon density in the $SU(N)$ Skyrme model that are not trivial embeddings of $SU(2) $ into $SU(N)$. We used Euler angles decomposition for arbitrary $N$ and the generalized 
hedgehog Ansatz at finite Baryon density. The Skyrmions of high topological charge that we find represent smooth Baryonic layers whose properties can be computed explicitly. In particular, we determine the energy to Baryon charge ratio for any $N$ 
showing the smoothness of the large $N$ limit. The closeness to the BPS bound of these configurations can also be analyzed. The energy density profiles of these finite density Skyrmions have \textit{lasagna-like} shape in agreement with recent 
experimental findings. The shear modulus can be precisely estimated as well and our analytical result is close to recent numerical studies in the literature.
\end{abstract}

\keywords{nuclear pasta, large N, skyrmions, multi-solitons}

\maketitle


\section{Introduction}

The characterization of the phase diagram of the low energy limit of QCD at finite Baryon density and low temperatures has motivated intense research in the last two decades, see \cite{Brambilla:2014jmp} and references therein. Analytic models are scarce and new exact results are hard to obtain. A well known example is the (3+1)-dimensional Nambu-Jona-Lasinio (NJL) model that shares some of the analytical difficulties of the low energy limit of QCD (see \cite{Nambu:1961tp} for a review). Together with the uselessness of perturbation theory at low energy, this means that the complicated phase diagram of low energy QCD cannot be easily analyzed with the available analytic techniques (see \cite{Rajagopal:2000wf, Alford:2000ze, Casalbuoni:2003wh} and references therein).

A remarkable feature of low energy QCD at finite Baryon density is that at low temperature very complex structures appear. When the Baryon density is increased, a phase that is commonly defined as \textit{nuclear pasta} appears. In \cite{Ravenhall:1983uh,pasta2,Horowitz:2014xca,Horowitz:2015gda,Dorso:2018lkv,daSilvaSchneider:2018yby,Caplan:2018gkr,Nandi:2018czr}, the presence of ``baryonic layers'' was disclosed, which will be the main focus of the present paper. Such a name arises from the fact that most of the baryonic charge and energy density is concentrated within lasagna-shaped regions in three dimensions\footnote{\textit{Nuclear spaghetti} and \textit{nuclear gnocchi} phases are also known to appear: see the references quoted above.}. Many physical properties of these configurations are currently under investigation, such as the elasticity of nuclear pasta and their transport properties \cite{Dorso:2018lkv,daSilvaSchneider:2018yby,Caplan:2018gkr,Nandi:2018czr}. The high topological charge of nuclear pasta makes it hard to study analytically.

As powerful numerical techniques are available to analyze these configurations (see, for instance, \cite{Dorso:2018lkv,daSilvaSchneider:2018yby,Caplan:2018gkr,Nandi:2018czr} and references therein), why should one insist so much in finding analytic solutions?
There are many reasons to strive for analytic solutions even when numerical techniques are available. \textit{Firstly}, it could be enough to remind all the fundamental concepts that we have understood thanks to the availability of the Kerr solutions in General Relativity and of the non-Abelian monopoles and instantons in Yang-Mills-Higgs theory. 
\textit{Secondly}, as in the present case, analytic solutions can disclose relevant physical properties of very complex structures which are difficult to analyze even numerically.

Until recently, these types of non-homogeneous condensates in the low energy limit of QCD in (3+1)-dimensions could not be properly understood analytically. A further problem is that, computationally, the large $N_f$ and large $N_c$ limits must be addressed carefully \cite{McLerran:2007qj,*Hidaka:2008yy, Glozman:2007tv}. One of the goals of the present paper is to shed light on the large $N_f$ behavior of these complex structures.

A simplified version of the low energy limit of QCD that encodes many relevant features is the (1+1)-dimensional version of the NJL model, also known as chiral Gross-Neveu model \cite{Gross:1974jv,Dashen:1975xh,Shei:1976mn,Feinberg:1996gz}. Such a model possesses a crystalline phase at low temperature and finite Baryon density \cite{Basar:2008ki,Schon:2000he,Thies:2006ti,Basar:2009fg,Torrieri:2010gz,Lottini:2012as,Lottini:2011zp,Torrieri:2011dg}. These results suggest that ordered structures must also appear in the low energy limit of QCD. At leading order in the \'{}t Hooft expansions \cite{tHooft:1974pnl,Veneziano:1976wm,Witten:1979kh,*Witten:1983tx}, the low energy limit of QCD is described by the Skyrme theory \cite{Skyrme:1961vq,*Skyrme:1961vr,*Skyrme:1962vh} (see \cite{Manton:2004tk,Shifman:2009zz} for reviews). Despite the Bosonic nature of the Skyrmion field $U$, its solitons represent Baryons (see \cite{Witten:1979kh,*Witten:1983tx,Finkelstein:1968hy,Balachandran:1982ty,Balachandran:1983dj, *Balachandran:1985fb,Adkins:1983ya}).

Here, we will analyze the appearance of complex structures at finite Baryon density in the $SU(N)$ Skyrme model in (3+1)-dimensions. We will focus on the analytic computations of relevant physical properties, such as the energy density, the energy per Baryon and the shear modulus of nuclear-lasagna like structures living at finite density \footnote{Pioneering works on the Skyrme model at finite density are \cite{Brey:1995zz,Klebanov:1985qi,Wuest:1987rc,Manton:1987xf,Goldhaber:1987pb}\ and references therein.}. We will compute their corresponding scaling with $N$.\\
We will combine the use of Euler angles for $SU(N)$ developed in \cite{Bertini:2005rc,Cacciatori:2012qi,Tilma:2004kp} together with the use of non-spherical hedgehog Ansatz introduced in (see \cite{Canfora:2013xja,*Canfora:2018rdz,Chen:2013qha,Canfora:2015xra,Ayon-Beato:2015eca,Alvarez:2017cjm,Aviles:2017hro,Canfora:2018clt,Canfora:2019asc,Canfora:2019kzd}).

\section{Skyrme action}

The action of the Skyrme model in four dimensions is 
\begin{equation}
S=\frac{K}{4}\int d^{4}x\sqrt{-g} \ \mathrm{tr}\left( R^{\mu }R_{\mu }+\frac{\lambda }{8}F_{\mu\nu}F^{\mu\nu}\right) \,,
\label{sky1}
\end{equation}
where $R_{\mu }=U^{-1}\partial _{\mu }U=R_{\mu }^{j}t_{j}$ with $U\in SU(N)$, $t_{i}$ the $SU(N)$ generators, and where $K$ and $\lambda \ $are the Skyrme couplings, $g$ is the the metric determinant,\footnote{We remind the reader that the $N$ of the $SU(N)$ of the Skyrme model corresponds to \textbf{N}$_{f}$.} and $F_{\mu \nu }=[R_{\mu },R_{\nu }]$. The field equations are
\begin{equation}
\nabla ^{\mu }\left( R_{\mu }+\frac{\lambda }{4}\left[ R^{\nu },F_{\mu\nu} \right] \right) =0\,.  \label{sky2.6}
\end{equation}
We will construct topologically non-trivial solutions at finite Baryon density. Our main goal is to determine the scaling with $N$ of relevant physical quantities. As we want to analyze Skyrmions of high topological charge living in flat spaces at finite Baryon density, we will consider the following metric: 
\begin{equation}
ds^{2}=-dt^{2}+L_{r}^{2}dr^{2}+L_{\gamma }^{2}d\gamma ^{2}+L_{\varphi}^{2}d\varphi ^{2}\ ,  \label{Ansatz-metric}
\end{equation}
while the range of coordinates is 
\begin{equation}
0\leq r\leq 2\pi \ ,\ 0\leq \gamma \leq 2\pi \ ,\ 0\leq \varphi \leq 2\pi \ .
\end{equation}
with the caveat that, despite the chosen values, \textit{they are not periodic!} The parameters \ $L_{r}$, $L_{\gamma }$ and $L_{\varphi }$ represent the size of the box within which the Skyrmion is confined.

\subsection{Quantities of high physical interest.}

\textit{Firstly}, the main goal of the paper is to compute the energy per Baryon and its large
N behavior. Therefore, only solutions with non-vanishing Baryon charge have
been considered. The usual definition of Baryon charge in the Skyrme model
(see \cite{Skyrme:1961vq,*Skyrme:1961vr,*Skyrme:1962vh,Witten:1979kh,*Witten:1983tx,Balachandran:1982ty,Balachandran:1983dj,*Balachandran:1985fb}) is%
\begin{equation}
W=B=\frac{1}{24\pi ^{2}}\int_{\left\{ t=const\right\} }\rho _{B}\ ,
\label{rational4}
\end{equation}%
\begin{equation}
\rho _{B}=\epsilon ^{ijk} \mathrm{tr}\left( U^{-1}\partial _{i}U\right) \left(
U^{-1}\partial _{j}U\right) \left( U^{-1}\partial _{k}U\right) \ ,
\label{rational4.1}
\end{equation}
so that a necessary condition in order to have non-trivial topological
charge is%
\begin{equation}
\rho _{B}\neq 0\ .  \label{necscond}
\end{equation}%
From the geometrical point of view, the above condition can be interpreted as
saying that the Skyrmion ``fills a three-dimensional spatial volume", at least
locally. On the other hand, such a condition is not sufficient in general.
One also has to require that the spatial integral of $\rho _{B}$ be a 
\textit{non-vanishing integer}:%
\begin{equation}
\frac{1}{24\pi ^{2}}\int_{\left\{ t=const\right\} }\rho _{B}\in 
\mathbb{Z}
\ .  \label{necscond2}
\end{equation}%
Usually, this second requirement allows to fix some of the parameters and
integration constants of the Ansatz, as we will see in the following.
However, there are more global conditions to be satisfied, as it will be
explained below. Hence, in the following we will only consider solutions
satisfying both the condition in Eq. (\ref{necscond}) and the one in Eq. (%
\ref{necscond2}).

\textit{Secondly}, the energy density (the $0-0$ component of the
energy-momentum tensor) reads 
\begin{align}
T_{00}=&-\frac{K}{2}\mathrm{tr}\left[ R_{0}R_{0}-\frac{1}{2}g_{00}R^{\alpha
}R_{\alpha }\right.\nonumber\\
&\left.+\frac{\lambda }{4}\left( g^{\alpha \beta }F_{0\alpha
}F_{0\beta }-\frac{g_{00}}{4}F_{\sigma \rho }F^{\sigma \rho }\right) \right]
\ ,  \label{timunu1}
\end{align}%
where $F_{\mu\nu}=[R_\mu,R_\nu]$. Thus, the total energy $E$ of the Skyrmion
is the spatial integral of the above quantity%
\begin{equation*}
E=\int_{\left\{ t=const\right\} }\sqrt{-g} \, T_{00}\ .
\end{equation*}
We \textit{define a Skyrmion }$U$\textit{\ to be static if its energy
density defined above is static}. In other words, a Skyrmion is static, if it
corresponds to a static distribution of energy density. It is worth to note
that this definition is more general than the naive definition of a static
Skyrmion as a static $SU(N)$-valued configuration $U$ which does not depend
on time. In particular an elegant approach to avoid Derrick's famous no-go
theorem on the existence of solitons corresponds to search for a
time-periodic Ansatz such that the energy density of the configuration is
still static, as it happens for boson stars \cite{Kaup:1968zz} (in the simpler case of 
$U(1)$-charged scalar field: see \cite{Liebling:2012fv} and references therein). The
Ansatz to be defined in the next sections will have exactly this property.
Moreover, unlike what happens for the usual Bosons star Ansatz for $U(1)$%
-charged scalar fields, the present Ansatz for $SU(N)$-valued scalar field
also possesses a non-trivial topological charge. Thus, we are interested in
solutions in which the energy density has non-trivial local maxima, which
could be identified with the position of the Skyrmions.

Given a solution of $SU(N)$ with Baryonic charge $B$ and energy $E$ living
in the metric (\ref{Ansatz-metric}) we have already mentioned, it is very interesting to analyze
the following quantity (which is nothing but the energy per Baryon of the
configuration $g\left( N,a\right) $)%
\begin{equation}
\frac{E}{B}\overset{def}{=}g\left( N,a\right) \ ,  \label{defratio}
\end{equation}%
where $a$ is any set of integration constants which characterizes the given
solution. It is especially interesting to understand the behavior of $%
g\left( N,a\right) $ defined above when $N$ is large (the \'{}t Hooft limit). Here and in the following we will call the function $g\left( N,a\right)$ the \textquotedblleft \textit{g-factor\textquotedblright }.
The very deep question is whether or not, in the given family of solutions
one is considering, one can define%
\begin{equation}
g^{\ast }\left( a\right) =\underset{N\rightarrow \infty }{\lim }g\left(
N,a\right)   \label{limitratio}
\end{equation}%
and if this limit is well defined. In particular, one would like to know
whether or not ``the closeness to the BPS bound" improves when $N$ is large.
Indeed, it is worth to remind that in the $SU(2)$ case all the known
solutions with non-vanishing topological charge exceed the bound by at least
the 20\%. Hence, one would like to know whether, in the \'{}t Hooft limit, the
``closeness of Skyrmions to the BPS bound" is finite or whether it grows without
bound. This issue is deeply related with the so-called Veneziano limit \cite%
{Veneziano:1976wm}, which is a variant of the \'{}t Hooft limit in which also the flavour 
number $N_{f}$ goes to infinity in such a way that $N_c/N_f$ stays finite. The Veneziano limit allows to take into
account the effects of quarks while keeping the advantages of the \'{}t Hooft
topological expansion. Since, to arrive at the Skyrme model as an effective
low energy limit of QCD, $N_{c}$ must be already large, the large $N$ limit which we are considering here (in which $N$ is
the one of the $SU(N)$ Skyrme model), can be considered as a sort
of Veneziano limit applied to the Skyrme model itself. The fact that such a
limit is smooth is a very non-trivial result which would be very difficult
to prove directly on the QCD Lagrangian.

The above discussion clearly shows that in order to declare a solution of
the Skyrme field equations as "physically interesting" two criteria must be
satisfied:

\textbf{1)} \textit{The topological charge of the solution must be
non-vanishing}

\textbf{2)} \textit{The energy density }$T_{00}$\textit{\ as function of the
coordinates must have an interesting pattern.}


\section{Local solutions}

Using the Euler angles for $SU(N)$ determined in \cite{Bertini:2005rc,Cacciatori:2012qi} together with the Ansatz for non-spherical Skyrmions living at finite Baryon density 
in \cite{Canfora:2013xja,*Canfora:2018rdz,Chen:2013qha,Canfora:2015xra,Ayon-Beato:2015eca,Alvarez:2017cjm,Aviles:2017hro,Canfora:2018clt,Canfora:2019asc,Canfora:2019kzd}, one arrives at the following Ansatz for the $SU(N)$ Skyrmion: 
\begin{align}
U[t,r,\varphi ,\gamma ]& =e^{\Phi k}e^{h(r)}e^{m\gamma k},  \label{Ansatz-FP}
\\
\Phi &=\frac{t}{L_{\varphi }}-\varphi \ ,  \label{Ansatz-FP1}
\end{align}
with a suitable choice of $k$ in $\mathfrak{su}(N)$ and $h(r)$ in the Cartan subalgebra 
$H$ to be specified below, $m$ a nonvanishing integer number, and where we
recall that the metric is given by (\ref{Ansatz-metric}).
When necessary to expand with respect to the basis of $\mathfrak{su}(N)$, we will also
write 
\begin{align}
h(r)=y_{1}(r)J_{1}+\ldots+y_{N-1}(r)J_{N-1},  \label{acca}
\end{align}
with (see app. \ref{app:su(n)}) 
\begin{align}
J_k=i(E_{k,k}-E_{k+1,k+1})\ , \quad k=1,\ldots,N-1.
\end{align}
In general we will use the simplifying notations
\begin{align}
 h'=\frac {d}{dr} h(r), \qquad h''=\frac {d^2}{dr^2} h(r).
\end{align}
As for $k$, for $c_j$ arbitrary complex numbers, forming the components of
the vector $\underline c\in \mathbb{C}^{N-1}$, we choose 
\begin{align}
k\equiv k_{\underline c}=\sum_{j=1}^{N-1} (c_j \lambda_j -c^*_j
\lambda_j^\dagger)\ ,  \label{kappa}
\end{align}
$\lambda_j\equiv \lambda_{\alpha_j}$ being the eigenmatrices of the simple
roots (App.\ref{app:su(n)}). We get 

\begin{prop}
\label{prop:eqmotion} From the Ansatz (\ref{Ansatz-FP}), (\ref{Ansatz-FP1}),
(\ref{Ansatz-metric}), the equations of motion reduce to 
\begin{align}
h^{\prime \prime }=\frac {\lambda m^2}{4L_\gamma^2} \left( [k,[k,h^{\prime
\prime }]]-[k,[h^{\prime },[h^{\prime },k]]] \right)\ ,  \label{eqmotion}
\end{align}
where the prime indicates derivation w.r.t. $r$.
\end{prop}

The proof is given in appendix \ref{app:prop1}. Exploiting (\ref{acca}) and (%
\ref{kappa}) we can further simplify the equations of motion, which can be
put in the following form.

\begin{align}
& h^{\prime \prime }+\frac{\lambda m^{2}}{2L_{\gamma }^{2}}%
\sum_{j=1}^{N-1}\alpha _{j}(h^{\prime \prime })|c_{j}|^{2}J_{j}=0,
\label{ridotta} 
\end{align}
\begin{widetext}
\begin{align}
& \sum_{j<k}\left( \alpha _{j}(h^{\prime})^2-\alpha _{k}(h^{\prime})^2-i\left(
\alpha _{j}(h^{\prime \prime })-\alpha _{k}(h^{\prime \prime })\right)
\right) c_{j}c_{k}[\lambda _{j},\lambda _{k}]  \qquad-h.c.=0,  \label{vincolo}
\end{align}%
where $h.c.$ stays for Hermitian conjugate, and $\alpha_j$ are a suitable choice of simple roots of $SU(N)$, defined in App. \ref{sec: simple roots}. 
Indeed, using (\ref{khhk}) and (%
\ref{kkh}), we can rewrite (\ref{eqmotion}) as 
\begin{equation}
h^{\prime \prime }=\frac{\lambda m^{2}}{4L_{\gamma }^{2}}\left\{ \sum_{j<k}%
\left[ i(\alpha _{j}(h^{\prime \prime })-\alpha _{k}(h^{\prime \prime
}))-(\alpha _{j}(h^{\prime})^2-\alpha _{k}(h^{\prime})^2)\right]
c_{j}c_{k}[\lambda _{j},\lambda _{k}]-h.c.-2\sum_{j=1}^{N-1}\alpha
_{j}(h^{\prime \prime })|c_{j}|^{2}J_{j}\right\} . \label{equazionciona}
\end{equation}
\end{widetext}
Now we use general properties of simple roots. Since $\lambda _{j}$ are
eigenmatrices relative to simple roots, it happens that or $[\lambda
_{j},\lambda _{k}]=0$ or it is an eigenmatrix relative to a positive root.%
\footnote{%
That is a linear combination of simple roots with non negative integer
coefficients} Similar considerations follow for $\lambda _{j}^{\dagger }$
w.r.t. negative roots. It follows that none of these terms can lie in $H$,
so that projecting equation (\ref{equazionciona}) on $H$ we get ($h^{\prime
\prime }$ belongs in $H$ by definition) we get (\ref{ridotta}), while
projecting on the complement we get (\ref{vincolo}).
These equations could be expressed even more explicitly in components, by exploiting (\ref{acca}) and using that $\alpha_j(J_k)={C_{A_{N-1}}}_{j,k}$ are the components of 
the Cartan matrix of $SU(N)$, as defined in App. \ref{app:sftf}, so that $\alpha_j(h^{(n)})=2y_j^{(n)}-y_{j+1}^{(n)}-y_{j-1}^{(n)}$. However, such explicit expression is not necessary in order to get the general solution.


\subsection{Explicit solutions}
We want now to find all the solutions of the equations (\ref{ridotta}) and
(\ref{vincolo}). To these hand we will make use of some technical fact
explained in appendix \ref{app:sftf}. Let us first consider (\ref{vincolo}).
Using (\ref{lambdajlambdak}) it becomes 
\begin{widetext}
\begin{align}
\sum_{j=1}^{N-2} \left(\alpha_j(h^{\prime})^2-\alpha_{j+1}(h^{\prime})^2-i
\left(\alpha_j(h^{\prime \prime })-\alpha_{j+1}(h^{\prime \prime })\right)
\right) c_jc_{j+1} E_{j,j+2} -h.c. =0.
\end{align} 
\end{widetext}
We will assume $\underline c$ to be generic, with this meaning that all the $%
c_j$ are non zero. Since $E_{j,j+2}$, including their conjugates, are all
linearly independent, this gives 
\begin{align*}
\alpha_j(h^{\prime})^2-\alpha_{j+1}(h^{\prime})^2-i \left(\alpha_j(h^{\prime
\prime })-\alpha_{j+1}(h^{\prime \prime })\right)=0, \\ \quad j=1,\ldots, N-2.
\label{queste}
\end{align*}
Since $\alpha_j$ are real valued, we get also 
\begin{align}
\alpha_j(h^{\prime \prime })=\alpha_{j+1}(h^{\prime \prime }), \quad\
\alpha_j(h^{\prime})^2-\alpha_{j+1}(h^{\prime})^2=0,\cr \quad\ j=1,\ldots, N-2.
\end{align}
The firsts of these give 
\begin{align}
\alpha_j(h^{\prime \prime })=\alpha_{1} (h^{\prime \prime }), \quad\
j=2,\ldots,N-1.  \label{alpha hsec}
\end{align}
We have two possibilities. Or $h^{\prime \prime }=0$, or not. We will now
show that the second case leads to a contradiction. First, notice that if $%
h^{\prime \prime }\neq 0$ then it must be $\alpha_j(h^{\prime \prime })\neq 0
$ for at least one $j$ (since the $\alpha_j$ are linearly independent), so
that all $\alpha_j(h^{\prime \prime })$ are equal and different from zero.
From the second of (\ref{queste}) we have that there must exist signs $%
\varepsilon_j$ such that 
\begin{align}
\alpha_j(h^{\prime })=\varepsilon_j \alpha_1(h^{\prime }), \quad\ j={%
2,\ldots,N-1}.  \label{alpha hpr}
\end{align}
Deriving it w.r.t. $r$ must give (\ref{alpha hsec}), so that $\varepsilon_j=1
$ for all $j$, and we are left with the linear system of equations 
\begin{align}
\alpha_j(h^{\prime })=\alpha_1(h^{\prime }), \quad\ j={2,\ldots,N-1}.
\end{align}
Since the $\alpha_j$ are linearly independent (of rank $N-1$) this is a set
of $N-2$ linearly independent equations for $h^{\prime }\in H$. Since $H$ is 
$N-1$ dimensional, the space of solutions is one dimensional and the general
solution of it is 
\begin{align*}
h^{\prime }(r) =f(r) v,
\end{align*}
where $f$ is an arbitrary function and $v\in H$ is the unique matrix
satisfying $\alpha_j(v)=1$ for all $j$ (which we will compute later, for now
it is sufficient to know it exists). We now replace this solution in (\ref%
{ridotta}). We immediately get 
\begin{align*}
f^{\prime }(r)\left( v+ \frac {\lambda m^2}{2L^2_\gamma} \sum_{j=1}^{N-1}
|c_j|^2 J_j \right)=0.
\end{align*}
Since we have assumed $h^{\prime \prime }\neq 0$, we have $f^{\prime }\neq 0$
and, therefore, 
\begin{align*}
v= -\frac {\lambda m^2}{2L^2_\gamma} \sum_{j=1}^{N-1} |c_j|^2 J_j \ .
\end{align*}
After applying $\alpha_k$ to this equality, using that $\alpha_k(v)=1$ and
noticing that $\alpha_k(J_j)= {C_{A_{N-1}}}_{k,j}$ are the components of the
Cartan matrix, we get 
\begin{align*}
1= -\frac {\lambda m^2}{2L^2_\gamma} \sum_{j=1}^{N-1} {C_{A_{N-1}}}_{k,j}
|c_j|^2, \qquad j=1,\ldots, N-1
\end{align*}
This relation can be inverted easily: if we consider $1$ at varying $j$ to
be the components of a vector in $\mathbb{R}^{N-1}$, we can apply the
inverse Cartan matrix to both members, thus getting 
\begin{align*}
|c_j|^2=-\frac {2L^2_\gamma}{\lambda m^2} \sum_{k=1}^{N-1} {C^{-1}_{A_{N-1}}}%
_{k,j}.
\end{align*}
Since $\lambda$ is positive and the same is true for the elements of the
inverse Cartan matrix (\ref{Cartinversa}), we see that this led us to a
contradiction. Therefore, the only possibility is that $f^{\prime }(r)=0$,
which is equivalent to $h^{\prime \prime }(r)=0$.\newline
Hence, we proceed in investigating the first possibility, $h^{\prime \prime
}=0$. In this case (\ref{ridotta}) is automatically satisfied and (\ref%
{vincolo}) reduces to (\ref{alpha hpr}). Its solution is 
\begin{align}
h^{\prime }(r)= a v
\end{align}
where $a$ is a constant and $v\in H$ is the unique matrix solving $%
\alpha_j(v)=\varepsilon_j$, $j=1,\ldots,N-1$ where $\varepsilon_j\in\{0,1\}$
(and $\varepsilon_1=1$). Since $\varepsilon_1$ is fixed, this gives $2^{N-2}$
solution for every choice of $c_j$ in $k$. As we will see in the explicit
example of $SU(4)$, however, not all of these are really distinct solutions.
There is a convenient way to express $v$ explicitly. Indeed, let us write%
\footnote{%
We omit an irrelevant additive integration constant} $h=a r v_\varepsilon$,
where $a$ is a constant and $v_\varepsilon \in H$ is a matrix 
\begin{align}
v_\varepsilon =\mathrm{diag} (v_1,\ldots,v_N) \label{vvarepsilon}
\end{align}
such that $\alpha_i (v_\varepsilon) =\varepsilon_i$, $\varepsilon_i=\pm 1$, $%
i=1,\ldots,N-1$ and of course $\sum_{i=1}^N v_i=0$. These equations are
easily solved by writing $v=\sum_{j=1}^{n-1}w_j J_j$ so that the equations are 
\begin{align*}
\varepsilon_k=\sum_{j=1}^{N-1} {C_{A_{N-1}}}_{k,j} w_j
\end{align*}
and the solution is 
\begin{align}
w_j =\sum_{k=1}^{N-1} {C^{-1}_{A_{N-1}}}_{j,k} \varepsilon_j
\end{align}
and 
\begin{align}
v_{\varepsilon}=\sum_{j,k} {C^{-1}_{A_{N-1}}}_{j,k} \varepsilon_k J_j.
\end{align}
We have thus proved 

\begin{prop}
All the solutions of the equations of motion (\ref{sky2.6}) determined by
the ans\"atze (\ref{Ansatz-FP}), (\ref{Ansatz-FP1}), (\ref{Ansatz-metric})
are given by 
\begin{align}
h(r)&=ar v_{\varepsilon}, \\
v_{\varepsilon}&=\sum_{j,k} {C^{-1}_{A_{N-1}}}_{j,k} \varepsilon_k J_j,
\label{vepsilon}
\end{align}
where $a$ is a real constant and $\varepsilon_j$ are signs, with $%
\varepsilon_1=1$.
\end{prop}

This solutions are only local solutions, which means that they solve the
differential equations. They do not extend automatically to global
solutions, that are solution with a well defined Baryon number. Looking for
global solutions is the task of the next section.


\section{Global solutions}

\label{sec:global} Up to now we have found the most general solution of the
differential Skyrme equation. Nevertheless, it is not sufficient to
determine a Skyrmion, since global conditions have to be imposed in order to
get a solution with a well defined topological charge. This condition is not
simply equivalent to impose that the topological charge must be integer
(this is just a consequence of the right topological condition) but that it
has to wrap a homological cycle an entire number of times (mathematically,
it has to cover a cycle, that means to be a surjective map with a well
defined degree). We will normalize the parametrizations so to have all
ranges in $[0,2\pi ]$. 

\subsection{Statement of the problem}

The difficulty in passing from local solutions to global solutions is
twofold. In order to illustrate it, let us consider the specific example of $%
SU(4)$ when $k$ is given by $c_i=1$. For getting a well defined global
solution, the function 
\begin{align}
g(\gamma)=e^{m\gamma k}
\end{align}
is expected to provide a good coordinate of the image, of the solution.
Since the target space of the map is compact, this requires that if we
extend the range of $\gamma$ to the whole $\mathbb{R}$, $g(\gamma)$ must
result to be a periodic function. Now, a simple calculation show that the
eigenvalues of $k$ are $\pm \mu_+, \pm \mu_-$, with 
\begin{align}
\mu_{\pm}= \frac i2 ( \sqrt 5 \pm 1).
\end{align}
This means that, for a suitable unitary constant matrix $U$, we have 
\begin{align}
g(\gamma)=U\mathrm{diag}(e^{m\gamma \mu_+}, e^{-m\gamma \mu_+}, e^{m\gamma
\mu_-}, e^{-m\gamma \mu_-}) U^\dagger.
\end{align}
In particular, its elements have periodicities $T_\pm$ with 
\begin{align}
T_{\pm}= \frac {2\pi}{m|\mu_\pm|}.
\end{align}
But since 
\begin{align}
\frac {T_+}{T_-} =\frac 12 (3+\sqrt 5)
\end{align}
is not rational, they have not a common period and the orbit never close, so
is not a periodic function but, rather, its orbit describes a curve which
densely covers a bi-torus in $SU(4)$. In particular, it is not possible to
use $g(\gamma)$ as a good factor to get a finite covering of a cycle,
despite it gives a solution of the equations of motion. It doesn't provide a
solution with a well defined topological number and must be discarded. One
has to tackle the problem of looking for acceptable matrices $k$, t hat are
matrices generating a well defined period.

\ 

Assuming we have solved the periodicity problem, there is a second subtlety
to be tackled: how to determine the right range of the coordinates in order
to correctly cover a cycle. First notice that $\pi_3(SU(N))=\mathbb{Z}$.
This suggests that homotopically we have just one representative for any
given topological (Baryonic) charge. Moreover, since $\pi_2(SU(N))=0$, we
have also $H_3(SU(N),\mathbb{Z})=\mathbb{Z}$, so we have also a unique
homological representative. Nevertheless, the solutions have not to be
identified under deformation, but at most under gauge equivalence. But since
the action is not gauge invariant, in our case all different representatives
in a given equivalence class must to be considered as different solutions.

We will distinguish three different classes of solutions. The first two
classes have canonical representatives: the ones of $SU(2)$-type, which
belong in every class, and the ones of $SO(3)$-type, which belong in even
classes only. They can be simply understood as follows. For any given $N$ we
can embed the representations of $su(2)$ into $su(N)$. Exponentiating, they
will give realisations of $SU(2)$ or $SO(3)$, depending on the specific
representation. This give rise to pure $SU(2)$-type or $SO(3)$-type
solutions. However, they can be continuously deformed, by varying the
corresponding $\underline{c}$ when allowed, giving rise to solutions that
are not embeddings, so we can consider them as \textit{true $SU(N)$ solutions%
}. But there exist a third class of solutions that cannot be obtained as
continuous deformations of embeddings. Their existence is due to the fact
that $SU(N)$ has a center isomorphic to $\mathbb{Z}_{N}$, which acts
continuously on $SU(N)$, see appendix \ref{app:su(n)}. In particular, if $%
\Gamma $ is a normal subgroup of the center, then one can construct the
group $SU(N)_{\Gamma }:=SU(N)/\Gamma $. The new class of solutions are
generated by cycles in $SU(N)$ that reduce to cycles of $SU(N)_{\Gamma }$
after the quotient. We will call them \textit{genuine $SU(N)$ solutions}. We
will consider them carefully in the explicit examples of $SU(3)$ and $SU(4)$%
, where everything is exactly computable, but now we shortly describe the $%
SU(2)$-type and $SO(3)$-type, where some details are a priory known, see
App. \ref{app:SO-SU}. \newline
An $SU(2)$-type cycle has the form 
\begin{equation*}
U(\phi ,\gamma ,\theta )=e^{\phi k}e^{h^{\prime }r}e^{m\gamma k},
\end{equation*}%
where $h^{\prime }$ is constant and the coordinate must run as follows. The
range of $r$ must be $T/4$, where $T$ is the period of $e^{h^{\prime }r}$.
The range of $\gamma $ must be $T_{k}$, the period of $e^{\gamma k}$ (with $%
m=1$!), and the range of $\phi $ must be $T_{k}/2$. 
Therefore, the convenient choice for the coordinates is 
\begin{equation*}
\varphi \in \lbrack 0,T_{k}/2],\quad r\in \lbrack 0,T/4],\quad \gamma \in
\lbrack 0,T_{k}]\ ,
\end{equation*}%
corresponding to the Baryon number 
\begin{equation*}
B=mB_{0},
\end{equation*}%
where $B_{0}$ is the fundamental charge of the given Skyrmion.\newline
For $SO(3)$-type cycles the interval for $\phi $ must cover an integer
period, so that the ranges must be 
\begin{equation*}
\varphi \in \lbrack 0,T_{k}],\quad r\in \lbrack 0,T/2],\quad \gamma \in
\lbrack 0,T_{k}]\ ,
\end{equation*}%
and the corresponding Baryon number is 
\begin{equation*}
B=2mB_{0}.
\end{equation*}The $SO(3)$-type can be defined as \textquotedblleft di-Baryon
class\textquotedblright\ after the seminal works \cite{Balachandran:1982ty} \cite{Balachandran:1983dj, *Balachandran:1985fb}.
These results were extended, keeping spherical symmetry, to the $SU(N)$ case
in \cite{Kopeliovich:1995mqa} \cite{Din:1980jg} \cite{Ioannidou:1999mk} \cite{Brihaye:2003qs} leading to numerical
non-embedded configurations in the $SU(N)$ Skyrme model. In the present
paper we will generalize those findings to the non-spherical case at finite
Baryon density achieving, moreover, analytic solutions. 


\subsection{$SU(3)$ Skyrmions}

Let us apply the above formalism to the case $N=3$. In this case we will see
that the problem of periodicity will not arise. 

\subsubsection{{$SO(3)$}-type solutions and genuine $SU(3)$ solutions}

The matrix $k$ is 
\begin{align}
k_{\underline c}= 
\begin{pmatrix}
0 & c_1 & 0 \\ 
-c_1^* & 0 & c_2 \\ 
0 & -c_2^* & 0%
\end{pmatrix}%
\ .
\end{align}
We put $\|\underline c\|^2=|c_1|^2+|c_2|^2$. Then, the characteristic
equation is 
\begin{align}
(\lambda^2+\|\underline c\|^2)\lambda=0.
\end{align}
The eigenvalues are $\lambda_0=0$ and $\lambda_\pm=\pm i\|\underline c \|$,
so that 
\begin{align}
g(\gamma) =e^{\gamma k_{\underline c}}
\end{align}
is periodic with period 
\begin{align}
T_k=\frac {2\pi}{\|\underline c\|}.
\end{align}
Now, we pass to determine the Cartan element. We have two possibilities
according to the two possible choices for $\underline \varepsilon$: 
\begin{align}
\underline \varepsilon_\pm = 
\begin{pmatrix}
1 \\ 
\pm 1%
\end{pmatrix}%
.
\end{align}
The inverse Cartan matrix for $SU(3)$ is 
\begin{align}
C_{A_2}^{-1}=\frac 13 
\begin{pmatrix}
2 & 1 \\ 
1 & 2%
\end{pmatrix}%
\ .
\end{align}
Thus, we find the two solutions 
\begin{align}
h_+(r)&= a r (J_1+J_2), \\
h_-(r)&= \frac a3 r (J_1-J_2).
\end{align}
The period of $\exp h_+(r)$ is 
\begin{align}
T_{h_+}=\frac {2\pi}a,
\end{align}
while the one of $\exp h_+(r)$ is 
\begin{align}
T_{h_-}=\frac {6\pi}a.
\end{align}
Now, we have to discuss the global properties in order to fix the ranges of
the parameters. To this end, accordingly to appendix \ref{app:SO-SU}, we
have to look for the intersection between the orbit of $h_\pm$ and the one
of $\gamma k_{\underline c}$. Using the characteristic equation we
immediately see that, see App. \ref{app:su3case}, 
\begin{align}
e^{\gamma k_{\underline c}} =I+\frac {\sin (\|\underline c\|\gamma)}{%
\|\underline c\|} k_{\underline c} +2\frac {\sin^2 (\frac {\|\underline
c\|}2 \gamma)}{\|\underline c\|^2} k^2_{\underline c}\ ,
\end{align}
so that the intersection we are looking for is just the unit matrix $I$.
However, we can notice that the orbit of $\exp h_-(r)$ contains the elements 
\begin{align}
\exp h_-(2\pi/a)=e^{\frac 23 \pi i} I, \qquad \exp h_-(4\pi/a)=e^{\frac 43
\pi i} I,
\end{align}
which are both in the center of $SU(3)$. Following App. \ref{app:SO-SU}, we
conclude that $h_-(r)$ defines a genuine $SU(3)$ solution, while only $h_+(r)
$ is of $SO(3)$-type.

In order to correctly define the solution we thus have to identify the
ranges as follows. First, it is
convenient to normalise $\underline{c}$ so that $\Vert \underline{c}\Vert =1$%
. This is just equivalent to re-scale the coordinates $\Phi $ and $\gamma $.
Therefore, we fix once for all the metric to be 
\begin{equation}
ds^{2}=-dt^{2}+L_{r}^{2}dr^{2}+L_{\gamma }^{2}d\gamma ^{2}+L_{\varphi
}^{2}d\varphi ^{2}\ ,  \label{Ansatz-metric-1}
\end{equation}%
with range of coordinates 
\begin{equation}
0\leq r\leq 2\pi \ ,\ 0\leq \gamma \leq 2\pi \ ,\ 0\leq \varphi \leq 2\pi \ ,
\end{equation}%
with the caveat that, despite the chosen values, \textit{none of the
coordinates is periodic!} Our Skyrmions are living in a rectangular box.%
\newline
\textbf{\boldmath{$SO(3)$} type solutions.} We already know that $r$ must
cover $1/2$ of the period of the Cartan torus, which implies that we have to
fix $a=\frac{1}{2}$. Hence, our solutions are 
\begin{align}
U_{\pm }^{\underline{c}}[t,r,\varphi ,\gamma ]& =e^{\Phi k_{\underline{c}%
}}e^{ar(J_{1}\pm J_{2})}e^{m\gamma k_{\underline{c}}},  \label{Ansatz-FP-0}
\\
\Phi & =\frac{t}{L_{\varphi }}-\varphi \ , \\
\varphi ,\gamma ,r& \in \lbrack 0,2\pi ], \\
B& =2m.
\end{align}%
More explicitly 
\begin{widetext}
\begin{equation*}
U_{+}^{\underline{c}}[t,r,\varphi ,\gamma ]=\left( I+\sin (\Phi )k _{%
\underline{c}}+2\sin ^{2}\frac{\Phi }{2}k _{\underline{c}}^{2}\right) 
\mathrm{diag}(e^{i\frac{r}{2}},1,e^{-i\frac{r}{2}})\left( I+\sin (m\gamma
)k _{\underline{c}}+2\sin ^{2}\frac{m\gamma }{2}k _{\underline{c}%
}^{2}\right) \ .
\end{equation*} 
\end{widetext}
We can now compute the energy end the factor $g_{+}=\frac{E}{2m}$. We omit
details here, since are particular cases of the general one for generic $N$
considered below. We get 
\begin{widetext}
\begin{equation*}
g_{+}(m,\underline{c})=L_{r}L_{\gamma }L_{\varphi }\frac{K\pi ^{3}}{m}\left[ 
\frac{4}{L_{\phi }^{2}}+\frac{1}{8L_{r}^{2}}+\frac{\lambda }{16L_{\phi
}^{2}L_{r}^{2}}+\frac{m^{2}}{L_{\gamma }^{2}}\left( 2+\frac{\lambda }{%
32L_{r}^{2}}+\frac{2\lambda }{L_{\phi }^{2}}(1-3|c_{1}|^{2}|c_{2}|^{2})%
\right) \right] ,
\end{equation*}%
where $|c_{1}|^{2}+|c_{2}|^{2}=1$. In particular, for each value of $m$, $%
|g_{+}(m,\underline{c})|$ takes its minimum at $|c_{1}|=|c_{2}|$, which is 
\begin{equation*}
g_{+}(m,\underline{c})=L_{r}L_{\gamma }L_{\varphi }\frac{K\pi ^{3}}{m}\left[ 
\frac{4}{L_{\phi }^{2}}+\frac{1}{8L_{r}^{2}}+\frac{\lambda }{16L_{\phi
}^{2}L_{r}^{2}}+\frac{m^{2}}{L_{\gamma }^{2}}\left( 2+\frac{\lambda }{%
32L_{r}^{2}}+\frac{\lambda }{2L_{\phi }^{2}}\right) \right] .
\end{equation*} 
\end{widetext}
Some comments are in order now. The reason for which the solution we have
just described are of $SO(3)$-type can be understood remembering that we are
working with $3\times 3$ matrices, which carry naturally a representation of
spin 1 of the rotation group. Indeed, the minimum energy case just
discussed, in which $|c_{j}|=1/\sqrt{2}$, corresponds exactly to the case
when the matrices $h_{+}$ and $k _{\underline{c}}$ are the generators
of the group $SO(3)$ in the representation of spin 1. The other solution,
for every fixed $m$, are continuous deformations obtained varying $%
\underline{c}$, which does not changes their topological nature, and in
particular the Baryon number, but it changes the energy. One can easily
check that for generic $\underline{c}$ the matrices $h_{+}$ and $k _{%
\underline{c}}$ do not generate a subgroup. One may wander if this is
related to the fact that their energy is not a minimum.\newline
The present remark suggests how to look for $SU(2)$-type solutions.\newline
\textbf{Genuine \boldmath$SU(3)$ type solutions.} Since this case does not
enter in the canonical classes, we have to manage separately the
determination of the correct ranges (then normalised to $2\pi $ as specified
above). As for $r$, we will prove in proposition \ref{prop:3} that in order
to have $r$ ranging in $[0,2\pi ]$, one has always to fix $a=\frac{1}{2}$.
For what concerns the other coordinates, let us notice that $h_{-}(r)$ does
not commute with $k_{\underline{c}}$ but it commutes with $k_{\underline{c}%
}^{2}$. Therefore, for $g(\gamma )=e^{\gamma k_{\underline{c}}}$, we see
that $g(T_{k}/2)$ commutes with $e^{h_{-}(r)}$. This means that we can write 
\begin{widetext}
\begin{equation*}
g(\Phi +T_{k}/2)e^{h_{-}(r)}g(\gamma )=g(\Phi
)g(T_{k}/2)e^{h_{-}(r)}g(\gamma )=g(\Phi )e^{h_{-}(r)}g(T_{k}/2)g(\gamma
)=g(\Phi )e^{h_{-}(r)}g(\gamma +T_{k}/2).
\end{equation*} 
\end{widetext}
If we assume that $U_{-}^{\underline{c}}[\Phi ,r,\gamma ]=g(\Phi
)e^{h_{-}(r)}g(\gamma )$ is covering a cycle, the relation $U_{-}^{%
\underline{c}}[\Phi +T_{k}/2,r,\gamma ]=U_{-}^{\underline{c}}[\Phi ,r,\gamma
+T_{k}/2]$ shows that we are covering it twice unless we restrict one of the
two ranges, of $\Phi $ and of $\gamma $, to one half the period of $g$. We
choose to reduce $\Phi $, so we replace $\Phi $ with $\Phi /2$. So, our
solution is 
\begin{align}
U_{-}^{\underline{c}}[t,r,\varphi ,\gamma ]& =e^{\frac{\Phi }{2}k_{%
\underline{c}}}e^{ar(J_{1}\pm J_{2})}e^{m\gamma k_{\underline{c}}}, \\
\Phi & =\frac{t}{L_{\varphi }}-\varphi \ , \\
\varphi ,\gamma ,r& \in \lbrack 0,2\pi ], \\
B& =m,
\end{align}%
where $B$ has been computed as in App. \ref{app:integrals}. Explicitly, 
\begin{widetext}
\begin{equation*}
U_{-}^{\underline{c}}[t,r,\varphi ,\gamma ]=\left( I+\sin \frac{\Phi }{2}%
k _{\underline{c}}+2\sin ^{2}\frac{\Phi }{4}k _{\underline{c}%
}^{2}\right) \mathrm{diag}(e^{i\frac{r}{6}},e^{-i\frac{r}{3}},e^{i\frac{r}{6}%
})\left( I+\sin (m\gamma )k _{\underline{c}}+2\sin ^{2}\frac{m\gamma }{2%
}k _{\underline{c}}^{2}\right) \ .
\end{equation*} 
For $U_{-}$, $g$ results to be independent from $\underline{c}$: 
\begin{equation*}
g_{-}(m,\underline{c})=L_{r}L_{\gamma }L_{\varphi }\frac{K\pi ^{3}}{2m}\left[
\frac{4}{L_{\phi }^{2}}+\frac{2}{3L_{r}^{2}}+\frac{\lambda }{4L_{\phi
}^{2}L_{r}^{2}}+8\frac{m^{2}}{L_{\gamma }^{2}}\left( 1+\frac{\lambda }{%
16L_{r}^{2}}+\frac{\lambda }{4L_{\phi }^{2}}\right) \right] .
\end{equation*}
\end{widetext}


\subsubsection{{$SU(2)$}-type solutions}

It is now clear that in order to find $SU(2)$-type solutions we have to
consider deformations of spin $\frac 12$ representations. This can be
obtained by ``reducing matrices'' down to $2\times 2$, and can be achieved
by choosing 
\begin{align}
k\equiv k_c = 
\begin{pmatrix}
0 & c & 0 \\ 
-c^* & 0 & 0 \\ 
0 & 0 & 0%
\end{pmatrix}%
\ ,
\end{align}
where $c$ is a phase. This is not the same thing as simply putting $c_2=0$
in $k_{\underline c}$ in the sense that we have to choose $k=k_c$ 
\emph{before} solving equation (\ref{vincolo}). Indeed, in (\ref{vincolo})
we assumed that all simple roots enter the game. This fixes the set of
possible choices of $h(r)$, and if in the above solutions we deform smoothly 
$\underline c$ to $(c,0)$, we cannot move away from our topological classes.
This is confirmed by the fact that if we put $c_2=0$, the matrix $k$ reduces
to a $2\times 2$ matrix, but the $k_\pm$ do not allow to reduce the
representation down to $\mathbb{C}^2$. We have to make a discontinuous
deformation. The point is that for $c_2=0$ the root $\alpha_2$ does not
enter into equation (\ref{vincolo}) that, indeed, for $N=3$ becomes just an
identity. This means that when $c_2=0$ we can choose for $h(r)$ any
combination 
\begin{align}
h(r)=a r J_1+ b r J_2,
\end{align}
with the only caveat that $e^{h(r)}$ must be periodic, so that $a$ and $b$
must be in rational ratio. We can set 
\begin{align}
h_q(r)= a r J_1+ a q r J_2, \qquad q\in \mathbb{Q}\ .
\end{align}
For $q=\pm 1$ we fall down to the previous $SO(3)$-solutions, while, of
course, $q=0$ provides a canonical embedding of $SU(2)$ into $SU(3)$, thus
identifying an $SU(2)$-type solution. It is worth to mention that, since $%
q\in \mathbb{Q}$ it cannot be deformed continuously among the three values,
compatibly with the fact that the case $q=0$ is not in the same topological
class of the other ones and, indeed, we may wander what happens for all the
other values of $q$, since they would generate new genuine $SU(3)$
solutions. However, it results that they have vanishing Baryon number, so
that we will not consider them further.\newline
Thus we get the solutions 
\begin{align}
U^c_0 [t,r,\varphi ,\gamma ]& =e^{\frac 12 \Phi k_{c}}e^{\frac r4
J_1}e^{n\gamma k_{c}},  \label{Ansatz-FP-1} \\
\Phi &=\frac{t}{L_{\varphi }}-\varphi \ , \\
\varphi,\gamma, r & \in [0,2\pi], \\
B&= n.
\end{align}
The $1/2$ factor in the first exponent has been added to ensure that when $%
\Phi$ varies in $[0,2\pi]$ it covers half of the period. Finally, we can
compute the factor $g$: 
\begin{widetext}
\begin{align}
g_0(n,c)=\frac {K\pi^3}n \left[ \frac 2{L_\phi^2} +\frac 1{4L_r^2} +\frac {%
\lambda}{8L_\phi^2 L_r^2} +\frac {n^2}{L_\gamma^2} \left( 4+\frac {\lambda}{%
4 L_r^2} +\frac {\lambda}{L_\phi^2} \right) \right].
\end{align}
\end{widetext}


\subsection{$SU(N)$ Skyrmions}

We will now consider the class of Skyrmions associated to the matrix $k$
given by 
\begin{align}
k_{\underline c}=\sum_{j=1}^{N-1} (c_j E_{j,j+1}-c^*_j E_{j+1,j}).
\end{align}
We will limit ourselves to the case when all the $c_j$ are different from
zero. Here, we have to face the problem of establish for which choices of $%
c_j$ the matrix $e^{\gamma k_{\underline c}}$ is periodic. By now, let
us assume to have solved it and write down the corresponding solution: 
\begin{align}
U^{\underline c}_\varepsilon [t,r,\varphi ,\gamma ]& =e^{\sigma \Phi
k_{\underline c}}e^{a v_\varepsilon r}e^{m\gamma k_{\underline c}},
\label{Ansatz-FP-2} \\
\Phi &=\frac{t}{L_{\varphi }}-\varphi \ , \\
\varphi,\gamma, r & \in [0,2\pi], \\
B&= \sigma 2m \|\underline c\|^2,
\end{align}
where $\sigma=1$ for $SO(3)$-type solutions and $\sigma=1/2$ for $SU(2)$%
-type solutions, and $v_\varepsilon$ is given by (\ref{vepsilon}). For
general genuine solutions the value of $\sigma$ must be computed case by
case. For any admissible $\underline c$ these are $2^{N-2}$ solutions (since 
$\varepsilon_1=1$). In principle $a$ could depend on $N$ and $\underline
\varepsilon$. However, we will now show that this is not the case and the
value of $a$ is completely fixed by requiring that the normalized interval $%
[0,2\pi]$ for $r$ must have the extension necessary to cover once a cycle:

\begin{prop}
\label{prop:3}  If $\exp (a v_{\underline \varepsilon} r)$ is such that $%
r\in [0,2\pi]$, and the corresponding map $U^{\underline c}_\varepsilon
[t,r,\varphi ,\gamma ]$ has not to cover a cycle more than once, then
necessarily $a=\frac 12$.
\end{prop}

\begin{proof}
The proof is simply based on the same strategy used for example in \cite%
{Bertini:2005rc}: one first constructs the invariant measure restricted to the
hypothetical cycle; the resulting measure will depend explicitly on some of
the coordinates and will vanish at specific value of that coordinate. The
good range for such a coordinate to cover just once a cycle is any range
between two vanishing points. The nice fact is now that the Haar measure
restricted to a cycle, a part from an eventual normalization constant, is
just $\rho_B$, which is computed in App.\ref{app:integrals}. Since it
results to depend on $r$ via $\sin (ar)$, we see that a suitable good
interval for $r$ is $[0,\pi/a]$. Since we want it to be $[0,2\pi]$, it must
be $a=\frac 12$.
\end{proof}

Therefore, we definitely have 
\begin{align}
a=\frac 12
\end{align}
in any case. Now, we can compute the $g$ factor for our solutions. To this
end, first note that 
\begin{widetext}
\begin{align}
T_{00}&=-\frac K2 \mathrm{Tr} \left( \frac 12 (R^\gamma
R_\gamma+R^rR_r)+\frac \lambda{16} F_{\rho\sigma} F^{\rho\sigma} +R_tR_t
+\frac \lambda4 g^{\alpha \beta} F_{t\alpha} F_{t\beta} \right)\cr & =
-\frac K4 \mathrm{Tr} \left( \frac {R_\gamma^2}{L_\gamma^2} +\frac {R_r^2}{%
L_r^2} \right)-\frac {K\lambda}{16} \mathrm{Tr} (F_{\gamma r})^2-\frac K2 
\mathrm{Tr} R_t^2- \frac {K\lambda}{8L_\varphi^2} \mathrm{Tr} \left( \frac {%
F_{\Phi r}^2}{L_r^2} +\frac {F_{\Phi\gamma}^2}{L_\gamma^2} \right).
\end{align} 
according to Appendix \ref{app:prop1}, and we used 
\begin{align}
R_{t}=\frac 1{L_\varphi} R_\Phi, \qquad F_{t\alpha}=\frac 1{L_\varphi}
F_{\Phi \alpha}.
\end{align}
According to (\ref{trakquadro}), (\ref{trhkhk}), (\ref{trxkxk}), with $%
a=\frac 12$, we have 
 \begin{align}
\mathrm{Tr} R_t^2=& \frac {\sigma^2}{L_\varphi^2} \mathrm{Tr} k_{\underline
c}^2=-\frac {2}{L_{\varphi}^2} \|\underline c\|^2 \sigma^2, \\
\mathrm{Tr} R_\gamma^2=& m^2 \mathrm{Tr} k_{\underline c}^2=-2m^2
\|\underline c \|^2, \\
\mathrm{Tr} R_r^2=& \frac 14 \mathrm{Tr} v_{\underline \varepsilon}^2=-\frac
14\sum_{j,k} {C^{-1}_{A_{N-1}}}_{j,k} \varepsilon_j \varepsilon_k\equiv
-\frac 14 \|v_{\underline \varepsilon}\|^2, \\
\mathrm{Tr} (F_{\gamma r})^2 =&m^2 \mathrm{Tr}([h^{\prime },k_{\underline
c}])^2= -\frac {m^2}2 \|\underline c\|^2, \\
\mathrm{Tr} (F_{\phi r})^2 =&\sigma^2 \mathrm{Tr}([x,h^{\prime}]^2)= -\frac {%
\sigma^2}2 \|\underline c\|^2, \\
\mathrm{Tr} (F_{\phi \gamma})^2 =\sigma^2 m^2\mathrm{Tr}([x,k_{\underline
c}])^2=& -8m^2\sigma^2 \sin^2 \frac r2 \left(\sum_{j=1}^{N-1}
|c_j|^4+\sum_{j=1}^{N-2} |c_j|^2 |c_{j+1}|^2 \frac 12(1-3\varepsilon_j
\varepsilon_{j+1}) \right).
\end{align}
Replacing in the expression for $T_{00}$ and using that the energy is 
\begin{align}
E=\int_0^{2\pi}dr\ \int_0^{2\pi}d\varphi\ \int_0^{2\pi} d\gamma\ L_r
L_\varphi L_\gamma T_{00}(r),
\end{align}
we get 
\begin{align}
E=& L_rL_\gamma L_\phi \|\underline c\|^2 \frac K2\pi^3 \left[ 16 \frac {%
\sigma^2}{L^2_\varphi} + \frac {\| v_{\underline\varepsilon}\|^2}{%
\|\underline c\|^2L_r^2}+\frac {\sigma^2 \lambda }{L^2_\varphi L^2_r}
\right. \cr & \left. + 8\frac {m^2 }{L_\gamma^2} \left( 1 + \frac {\lambda }{%
16 L_r^2} + \frac {\lambda \sigma^2}{L^2_{\varphi} \|\underline c\|^2}
\left(\sum_{j=1}^{N-1} |c_j|^4+\sum_{j=1}^{N-2} |c_j|^2 |c_{j+1}|^2
\left(\frac 12-\frac 32\varepsilon_j \varepsilon_{j+1}\right) \right)
\right) \right].
\end{align}
\end{widetext}
In a similar way one can compute the baryon number. This is done in appendix %
\ref{app:integrals}, with the result 
\begin{align}
B=2m\sigma \|\underline c \|^2.
\end{align}
From these results we immediately get the $g$-factor: 
\begin{widetext}
 \begin{align}  \label{formulone}
g(N,m,\underline c,\varepsilon)=& L_rL_\gamma L_\phi \frac {K\pi^3}{4\sigma m%
} \left[ 16 \frac {\sigma^2}{L^2_\varphi} + \frac {\|v_{\underline
\varepsilon}\|^2}{\|\underline c\|^2L_r^2}+\frac {\sigma^2 \lambda }{%
L^2_\varphi L^2_r} \right. \cr & \left. + 8\frac {m^2 }{L_\gamma^2} \left( 1
+ \frac {\lambda }{16 L_r^2} + \frac {\lambda \sigma^2}{L^2_{\varphi}
\|\underline c\|^2} \left(\sum_{j=1}^{N-1} |c_j|^4+\sum_{j=1}^{N-2} |c_j|^2
|c_{j+1}|^2 \left(\frac 12-\frac 32\varepsilon_j \varepsilon_{j+1}\right)
\right) \right) \right].
\end{align}
\end{widetext}
Up to now, we have assumed $\underline c$ to be normalised so that $%
g(\gamma)=e^{\gamma k_{\underline c}}$ has period $2\pi$. However, we
will not have really found a solution until we will be able to specify for
which $\underline c$ the function $g$ is periodic. Therefore, we cannot
further postpone to tackle this problem.\newline
However, before considering it in general, we want now concentrate on a very
particular case, when $\varepsilon_j=1$ for all $j$. In this case
\begin{widetext}
 \begin{align}
g(N,m,\underline c)=& L_rL_\gamma L_\phi\frac {K\pi^3}{4\sigma m} \left[ 16 
\frac {\sigma^2 }{L^2_\varphi} + \frac {\| v\|^2}{\|\underline c\|^2L_r^2}+%
\frac {\sigma^2 \lambda}{L^2_\varphi L^2_r} \right. \cr & \left. + 8\frac {%
m^2 }{L_\gamma^2} \left( 1 + \frac {\lambda }{16 L_r^2} + \frac {\lambda
\sigma^2}{L^2_{\varphi}\|\underline c\|^2} \left(\sum_{j=1}^{N-1}
|c_j|^4-\sum_{j=1}^{N-2} |c_j|^2 |c_{j+1}|^2 \right) \right) \right]\ .
\end{align}
\end{widetext} 
It is clear that, among all possible choices for $\varepsilon_j$, this
minimises the energy, apart from possible effects due to $\|v\|$. We want
also minimise with respect to the $c_j$, assuming the normalisation of $%
\|\underline c\|$ fixed. Introducing a Lagrange multiplicator $\Lambda$, we
have to extremize the function 
\begin{align}  \label{fc}
f(\underline c) =\sum_{j=1}^{N-1} |c_j|^4-\sum_{j=1}^{N-2} |c_j|^2
|c_{j+1}|^2 -\Lambda \|\underline c\|^2.
\end{align}
Deriving with respect to $|c_j|^2$ we get the system 
\begin{align}
C_{A_{N-1}} \underline {|c|^2} =\Lambda \underline 1,
\end{align}
$\underline 1$ being the vector in $\mathbb{R}^{N-1}$ having all elements
equal to 1. This gives the solution 
\begin{align}
|c_j|^2=\frac \Lambda2 j(N-j).  \label{cj}
\end{align}
Interestingly this also solves automatically the periodicity problem. It is
easy to see (App. \ref{periodicity}) that 
\begin{align}  \label{first-solution}
c_j =\zeta_j \sqrt {\frac \Lambda2 j(N-j)}, \quad \Lambda= 
\begin{cases}
\frac 12 & \mbox{ for odd } N \\ 
2 & \mbox{ for even } N%
\end{cases}%
\end{align}
where $\zeta_j$ are arbitrary phases, give a matrix $e^{\gamma
k_{\underline c}}$ that is periodic in $\gamma$ with period $2\pi$. For 
$v$ we find 
\begin{align}  \label{v}
v= \sum_{j,k} {C^{-1}_{A_{N-1}}}_{j,k} J_j.
\end{align}
Moreover, we have

\begin{prop}
If $c_j$ are given by (\ref{first-solution}), and $v$ is as in (\ref{v}),
then  
\begin{align}
\|\underline c\|^2& =\frac {\Lambda}{12} N(N^2-1), \\
\|v\|^2& =\frac {1}{12} N(N^2-1),
\end{align}
and 
\begin{eqnarray}
\sum_{j=1}^{N-1} |c_j|^4-\sum_{j=1}^{N-2} |c_j|^2 |c_{j+1}|^2&=&\frac {%
\Lambda^2}{24} N(N^2-1)\nonumber \\&=&\frac {\Lambda}2 \|\underline c\|^2.
\end{eqnarray}\end{prop}
{\it Proof.}
The first result follows immediately by the well known formulas 
\begin{align}
\sum_{j=1}^{N-1} j&=\frac {N(N-1)}2, \\ 
\sum_{j=1}^{N-1} j^2&=\frac {%
N(N-1)(2N-1)}6.
\end{align}
For the second expression notice that, by using (\ref{Cartinversa}), 
\begin{align}
\|v\|^2=& \sum_{j,k} {C^{-1}_{A_{N-1}}}_{j,k} \cr =&\frac 1N \left[ \sum_{j<k}
j(N-k)+\sum_{j\geq k} k(N-j) \right] \cr = & -\frac 1N
\sum_{j,k}jk+\sum_{j<k} j+\sum_{j\geq k} k\cr =&-\frac
1N(\sum_{j=1}^{N-1}j)^2+\sum_{j=1}^{N-1}j(N-j-1)\cr &+\sum_{k=1}^{N-1} k(N-k),
\end{align}
and the final expression again follows after applying the above well known
formulas. \newline
For the last formula, notice that the $c_j$ are solutions of 
\begin{align}
\frac {\partial f}{\partial c_k}=0, \quad\ k=1,\ldots,N-1,
\end{align}
where $f$, given by (\ref{fc}). From this we get 
\begin{align}
\sum_{k=1}^{N-1}c_k\frac {\partial f}{\partial c_k}=0.
\end{align}
Now, $f$ is the sum of two homogeneous pieces, one of degree 4 end the other
of degree 2. Therefore, we can use the Euler theorem\footnote{%
t.i. for an homogeneous function $f:\mathbb{R}^N\rightarrow \mathbb{R}$ of
degree $L$ one has 
\begin{align*}
\vec x \cdot \mathrm{grad}f= Lf.
\end{align*}%
} to rewrite the last as 
\begin{align}
0=4\left(\sum_{j=1}^{N-1} |c_j|^4-\sum_{j=1}^{N-2} |c_j|^2
|c_{j+1}|^2\right) -2\Lambda \|\underline c\|^2,
\end{align}
which completes the proof. $\qquad\qquad\ \Box$

Using this results and noticing that $\sigma ^{2}\Lambda =1/2$, we find for
the energy per Baryon 
\begin{widetext}
 \begin{align}
g(N,m)_{min}=& L_{r}L_{\gamma }L_{\varphi }\frac{K\pi ^{3}}{\sigma m}\left[ 4%
\frac{\sigma ^{2}}{L_{\varphi }^{2}}+\frac{\sigma ^{2}}{2L_{r}^{2}}+\frac{%
\sigma ^{2}\lambda }{4L_{\varphi }^{2}L_{r}^{2}}+2\frac{m^{2}}{L_{\gamma
}^{2}}\left( 1+\frac{\lambda }{16L_{r}^{2}}+\frac{\lambda }{4L_{\varphi }^{2}%
}\right) \right] \ ,  \label{gmin} \\
\Lambda =& 2^{(-1)^{N}},\quad \sigma =2^{-\frac{(-1)^{N}+1}{2}}, \\
B=& 2^{\frac{1+(-1)^{N}}{2}}m\frac{1}{12}N(N^{2}-1)\ .
\end{align}%
\end{widetext}
Notice that $g(N,m)$ depends on $N$ only through $\sigma $.\newline
We can also notice that 
\begin{equation*}
I_{N}=\frac{N(N^{2}-1)}{6}
\end{equation*}%
is the \textit{Dynkin index} of the given representation of the principal
representation of $sl(2)$ in $sl(N)$, so, the fundamental Baryonic charge
associated is 
\begin{equation*}
B=\frac{1}{2\sigma }I_{N}\ .
\end{equation*}%
Notice that for $N$ odd $I_{N}$ is even, so $B$ is always integer. \newline
Finally, we are also interested in minimizing expression (\ref{gmin}) with
respects to $L_{a}$, $a=\varphi ,r,\gamma $. This is done in general in
appendix \ref{app:minimal energy}. By using the formulas therein and the
ones in the last proposition, we get that the minimum is reached at 
\begin{equation*}
L_{\varphi }=\frac{\sqrt{\lambda }}{2^{\frac{3}{4}}}\qquad L_{r}=\frac{\sqrt{%
\lambda }}{4}\qquad L_{\gamma }=\frac{m}{\sigma }\frac{\sqrt{\lambda }}{2^{%
\frac{5}{4}}}, \label{lati}
\end{equation*}%
with corresponding minimal value 
\begin{equation}
g_{min}=K\sqrt{\lambda }\pi ^{3}(1+2\sqrt{2}).\label{g-min-dim}
\end{equation}%
Using normalised units (corresponding to $\lambda =1$ and $K=(6\pi ^{2})^{-1}
$) we get 
\begin{equation}
g_{min,stand}=\pi \frac{1+2\sqrt{2}}{6}\approx 2.00456. \label{minimal g}
\end{equation}%
Notice that this is independent from $N$ and it is expected to be the
absolute minimum with respect to any choice of $\varepsilon _{j}$. We will
not try to prove this conjecture here, we will limit ourselves to check it
for $N=4$ here below. The comparison with \cite{Sutcliffe:2010et}\ is very
interesting. The present results are slightly above the bound in \cite%
{Sutcliffe:2010et} due to the time-dependence in the Ansatz. Note however that the
present time-dependence cannot be undone as the present solutions wrap in a
topologically non-trivial way also around the time direction. 
To the best of our
knowledge, this is the first analytic computations showing explicitly how
the \textit{closeness to the BPS bound} ``evolves'' with $N$ in the $SU(N)$ Skyrme model.

To be more specific, as it has been already emphasized, we are interested in
topologically non-trivial solutions. In the present context this means that
we only consider $SU(N)$ Ansatz such that%
\begin{equation*}
\rho _{B}=Tr\left( U^{-1}dU\right) ^{3}\ \neq 0\ .
\end{equation*}%
As it has been discussed in the \mbox{previous sections}, $\rho _{B}$ represents the Baryon density when it is non-vanishing along three-dimensional
space-like hyper\-surfaces $\Sigma _{t=const}$.

In these cases, the integral of $\rho _{B}$ over $\Sigma _{t=const}$ represents the Baryon charge. While,
mathematically, these integrals represent how many time the $SU(N)$-valued
Skyrmions wrap around $\Sigma _{t=const}$. On the other hand, $\rho _{B}$
can be topologically trivial also along time-like hypersurfaces. In this
case, one can also consider the wrapping of the $SU(N)$-valued
configurations along three-dimensional time-like~hypersurfaces. The
configurations which have been constructed here are, as a direct check
easily reveals, topologically non-trivial in two ways. Not only they possess
non-vanishing Baryonic charge, they are also wrapped non-trivially along
time-like hypersurfaces. Indeed, if one considers%
\begin{align*}
U_{\varepsilon }^{\underline{c}}[t,r,\varphi ,\gamma ]& =e^{\sigma \Phi
k _{\underline{c}}}e^{av_{\varepsilon }r}e^{m\gamma k _{\underline{%
c}}}, \\
\Phi & =\frac{t}{L_{\varphi }}-\varphi \ ,
\end{align*}%
then the corresponding topological density has one space-like component and
one time-like component:%
\begin{equation*}
\rho _{B}\sim dr\wedge d\varphi \wedge d\gamma -dr\wedge d\left( \frac{t}{L}%
\right) \wedge d\gamma \ .
\end{equation*}%
In particular, it implies that these $SU(N)$ Skyrmions wrap non-trivially
around the three-dimensional time-like $\left\{ \varphi =const\right\} $
hypersurfaces. The consequence of this fact is that the time-dependence of
the present configuration ``cannot be undone'' otherwise the winding number
corresponding to the $\left\{ \varphi =const\right\} $ hypersurfaces would
change.

\subsection{Solving the periodicity problem} The solution of this problem is
provided in App.\ref{app:solve the problem}.\footnote{S.L.C. is particularly grateful to Laurent Lafforgue for suggesting him how to tackle this problem in full generality.} 
We discuss here the main results. The vectors $\underline c\in \mathbb{C}^{N-1}$ having all
components different from zero and allowing for a periodic function $%
g(\gamma)=e^{\gamma k_{\underline c}}$, with period $2\pi$, form a family 
\begin{align}
\underline c =\underline c (\underline m, \underline \alpha, \underline t ),
\end{align}
where $\underline m=(m_1,\ldots,m_n)$, is a finite strictly increasing
sequence of strictly positive coprime integer numbers, $n$ is the integer
part of $N/2$, $\underline \alpha\in [0,2\pi)^{N-1}$, and $\underline t \in
W\subset \mathbb{R}^{N-n-1}$ is a set of parameters parametrizing the
strictly positive real solutions of the algebraic system 
\begin{align}
\sum_{j=1}^{N-1} \zeta_j=& \sum_{a=1}^n m_a^2, \\
\sum_{j_1\ll\ldots \ll j_k\leq {N-1}} \!\!\!\!\!\!\!\!\! \zeta_{j_1} \cdots
\zeta_{j_k}=&\!\!\!\! \!\!\! \sum_{a_1<\ldots<a_k\leq n}\!\!\!\!\!\!\! m_{a_1}^2\cdots
m_{a_k}^2, \cr k=2,\ldots,n,
\end{align}
in the real variables $\zeta_j$, $j=1,\ldots, N-1$.\newline
The parameters $\underline \alpha$ and $\underline t$ form a moduli space $%
\mathbb{T}^{N-1} \times W$. The relevant physical quantities depend only on $%
|c_j|$ so are independent on the components in $N-1$ dimensional torus.
Therefore, we can say that only $W$ represents the \textit{relevant moduli}.
As one could expect, in particular, the Baryon number associated to a
solution constructed with $\underline c (\underline m, \underline \alpha,
\underline t )$ depends only on $\underline m$ and not on the continuous
moduli: 
\begin{align}
B=2\sigma m \sum_{a=1}^n m_a^2.
\end{align}
The general form of $g(\gamma)$ is 
\begin{align}
e^{\gamma k_{\underline c}}=f_0(\gamma,\underline m) \mathbb{I}%
+\sum_{j=1}^{N-1} f_j(\gamma,\underline m) k_{\underline c (\underline m,
\underline \alpha, \underline t )}^j,
\end{align}
where the $f_\beta$, $\beta=0,\ldots,N-1$ are linear combinations of $1$ and 
$\sin(m_a\gamma)$, $\cos(m_a \gamma)$, with rational functions of $%
\underline m$ as coefficients, and satisfying $f_0(0,\underline m)=1$, $%
f_j(0,\underline m)=0$ for $j>0$. In particular, the dependence on the
continuous moduli is only through the $k_{\underline c}^j$.


\subsection{Back to $N=4$}

Following App.\ref{sec:N=4}, for any two coprime positive integers $p$ and $q
$ such that $p> q$, for $N=4$ we can find four families of solutions, each
one parametrized by three real phases $\alpha_1, \alpha_2, \alpha_3$ and a
real modulus $\tau\in [q,p]$. Each of these families is specified by one of
the four possible inequivalent choices for the discrete vector $\varepsilon$. 
Recall that in this case the inverse Cartan matrix is 
\begin{align}
C^{-1}_{A_3}=\frac 14 
\begin{pmatrix}
3 & 2 & 1 \\ 
2 & 4 & 2 \\ 
1 & 2 & 3%
\end{pmatrix}%
.
\end{align}
We also have, see App.\ref{sec:N=4}, 
\begin{widetext}
\begin{align}
g_4(x)\equiv e^{x k_{\underline c}} =& \left(\frac {p^2}{p^2-q^2} \cos (qx)-%
\frac {q^2}{p^2-q^2} \cos (px)\right) \mathbb{I }+ \left(\frac {p^2}{%
q(p^2-q^2)} \sin (qx)-\frac {q^2}{p(p^2-q^2)} \sin (px)\right) k_{\underline
c}\cr & +\left(\frac {1}{p^2-q^2} (\cos (qx)-\cos (px))\right) k_{\underline
c}^2+\left( \frac {1}{p^2-q^2} \left( \frac {\sin (qx)}q-\frac {\sin (px)}p
\right) \right) k_{\underline c}^3,
\end{align} 
with{\ 
\begin{align}
k_{\underline c}=& 
\begin{pmatrix}
0 & e^{i\alpha_1} \tau & 0 & 0 \\ 
-e^{-i\alpha_1} \tau & 0 & e^{i\alpha_2} \psi & 0 \\ 
0 & -e^{-i\alpha_2} \psi & 0 & \frac {pq}\tau e^{i\alpha_3} \\ 
0 & 0 & -\frac {pq}\tau e^{-i\alpha_3} & 0%
\end{pmatrix}%
, \\
k_{\underline c}^2=& 
\begin{pmatrix}
-\tau^2 & 0 & e^{i(\alpha_1+\alpha_2)} \tau\psi & 0 \\ 
0 & -(p^2+q^2 -\frac {p^2q^2}{\tau^2}) & 0 & e^{i(\alpha_2+\alpha_3)}\frac {%
pq}\tau \psi \\ 
e^{-i(\alpha_1+\alpha_2)} \tau\psi & 0 & -(p^2+q^2-\tau^2) & 0 \\ 
0 & e^{-i(\alpha_2+\alpha_3)}\frac {pq}\tau \psi & 0 & -\frac {p^2q^2}{\tau^2%
}%
\end{pmatrix}%
, \\
k_{\underline c}^3=&
\begin{pmatrix}
0 & -e^{i\alpha_1} \tau (p^2+q^2 -\frac {p^2q^2}{\tau^2}) & 0 & 
e^{i(\alpha_1+\alpha_2+\alpha_3)}pq \psi \\ 
e^{-i\alpha_1} \tau (p^2+q^2 -\frac {p^2q^2}{\tau^2}) & 0 & 
-e^{i\alpha_2}(p^2+q^2) \psi & 0 \\ 
0 & e^{-i\alpha_2}(p^2+q^2) \psi & 0 & -\frac {pq}\tau
e^{i\alpha_3}(p^2+q^2-\tau^2) \\ 
-e^{-i(\alpha_1+\alpha_2+\alpha_3)}pq \psi & 0 & \frac {pq}\tau
e^{-i\alpha_3}(p^2+q^2-\tau^2) & 0%
\end{pmatrix},
\end{align}
} and 
\begin{align}
\psi= \sqrt {p^2+q^2-\tau^2 -\frac {p^2q^2}{\tau^2}}.
\end{align}
\end{widetext}


\subsubsection{The almost $SU(2)$-type solutions}

\noindent The $SU(2)$ solution is expected to be identified by $%
\varepsilon_a =(1,1,1)$. Indeed, from (\ref{vepsilon}) we have 
\begin{align}
v_a=\frac i2\mathrm{diag}(3,1,-1,-3), \qquad \|v_a\|^2=5,
\end{align}
which is exactly the matrix representing the diagonal generator of $SU(2)$
in the spin $3/2$ representation. However, this is not true in general and
we will see that in this series only the one with $(p,q)=(3,1)$ is
deformable to an $SU(2)$ embedding. Let us first look at the coordinate
ranges. Regarding the range of $r$, it is completely fixed by
proposition \ref{prop:3}. As for the remaining ranges, they must correspond
to the period of $g_4$, unless there are (finite discrete) subgroups of the $%
U(1)$ group generated by $g_4$, which commute with $v_a$. Since $v_a$ does
not commute with $k_{\underline c}^j$, $j=1,2,3$ (or any linear combination
thereof), we have to look for the values of $x$, such that $f_j(x)=0$, $j=1,2,3
$ (App.~\ref{sec:N=4}). Looking at $f_3$, this means 
\begin{align}
\cos (px)=\cos (qx),
\end{align}
that is 
\begin{align}
px=\pm qx+2\ell \pi
\end{align}
for some integer $\ell$. For the $x$ satisfying this condition, call them $%
x_pm$, one has for $f_3$ 
\begin{align}
f_3(x\pm)=\frac 1{p^2-q^2} (\frac 1q-\frac 1p) \sin (qx_\pm),
\end{align}
which is zero for $x=j\pi$ for some integer $j$. Since our coordinates are
forced to vary in $[0,2\pi]$, the only non trivial possibility is $x=\pi$.
Putting this back into the previous condition, we must also have 
\begin{align}
p-q=2\ell,
\end{align}
which means that, since $p$ and $q$ are coprime, this happens only when
both $p$ and $q$ are odd. In this case $g_4(\pi)=-\mathbb{I}$. Therefore, we
see that for $p-q$ odd, there are no discrete symmetries, and the ranges of $%
\Phi$ and $\gamma$ must coincide with the whole period, so $\sigma=1$.
Instead, for $p-q$ even we have 
\begin{align}
&g_4(\Phi+\pi) e^{\frac 12 v_a r} g_4(\gamma+\pi) \cr &=g_4(\Phi) (-\mathbb{I}%
)e^{\frac 12 v_a r} (-\mathbb{I}) g_4(\gamma)\cr &=g_4(\Phi) e^{\frac 12 v_a r}
g_4(\gamma),
\end{align}
so we see that to any point on the image there correspond two different
coordinates $(\Phi, \gamma)$ and $(\Phi+\pi, \gamma+\pi)$, unless we restrict
one of the two ranges to half a period. We choose to do it with $\Phi$, and,
in order to keep its range to be $[0,2\pi]$, we fix $\sigma=1/2$.\newline
The field is 
\begin{align}
U_a=&g_4(\sigma_{p-q}\Phi) 
\begin{pmatrix}
e^{\frac 34 i r} & 0 & 0 & 0 \\ 
0 & e^{\frac 14 i r} & 0 & 0 \\ 
0 & 0 & e^{-\frac 14 i r} & 0 \\ 
0 & 0 & 0 & e^{-\frac 34 i r}%
\end{pmatrix} g_4(m\gamma), 
\end{align}
\begin{align}
\Phi=&\frac t{L_\varphi}-\varphi, \\
\sigma_{p-q}=& 
\begin{cases}
\frac 12 & \mbox{ if $p-q$ is even} \\ 
1 & \mbox{ if $p-q$ is odd }%
\end{cases}%
.
\end{align} 
The Baryon number is $B_a=2\sigma_{p-a}m(p^2+q^2)$, while for the $g$-factor
we get 
\begin{widetext}
\begin{align}
g_a(&p,q,m, \tau)= L_rL_\gamma L_\phi \frac {K\pi^3}{4\sigma_{p-q} m} \left[ 
\frac {16\sigma^2_{p-q}}{L^2_\varphi} + \frac {5}{(p^2+q^2)L_r^2}+\frac {%
\lambda \sigma^2_{p-q}}{L^2_\varphi L^2_r} + 8\frac {m^2 }{L_\gamma^2}
\left( 1 + \frac {\lambda }{16 L_r^2} \right) \right. \cr + & \left.
8\sigma^2_{p-q}\frac {m^2 }{L_\gamma^2} \frac {\lambda(p^2+q^2)}{%
L^2_{\varphi} } \left( 1-\frac {3\tau^2}{p^2+q^2} +\frac {3\tau^4}{%
(p^2+q^2)^2}+\frac {4p^2q^2}{(p^2+q^2)^2}+\frac {3p^4q^4}{(p^2+q^2)^2} \frac
1{\tau^4} -\frac {3p^2q^2}{p^2+q^2}\frac 1{\tau^2} \right) \right]\!\!.
\end{align} 
The corresponding minimal energy, expressed in normalized units, is given by
(\ref{energia minima}), which in this case becomes: 
\begin{align}
g_a(p,q,\tau)=\frac {\pi}{3\sqrt 2} \left[2+\sqrt 5 \left( 1-\frac {3\tau^2}{%
p^2+q^2} +\frac {3\tau^4}{(p^2+q^2)^2}+\frac {4p^2q^2}{(p^2+q^2)^2}+\frac {%
3p^4q^4}{(p^2+q^2)^2} \frac 1{\tau^4} -\frac {3p^2q^2}{p^2+q^2}\frac
1{\tau^2} \right)^{\frac 12} \right].
\end{align}
\end{widetext}
We can further minimise w.r.t. $\tau$. Setting $x=\tau^2$, we have to find
the stationary points in 
\begin{align}
q^2 < x< p^2.
\end{align}
Deriving the expression in the square root w.r.t. $x$ and multiplying by $%
(p^2+q^2)^2x^3/6$, we get the equation 
\begin{align}
0= \left( x^2-p^2q^2\right) \left(x^2 -\frac x2 (p^2+q^2) +p^2q^2 \right).
\end{align}
This gives the admissible solutions ($x$ is positive) 
\begin{align}
x_0=pq, \qquad x_\pm =\frac {p^2+q^2}4 \pm \sqrt {\frac {(p^2+q^2)^2}{16}%
-p^2q^2}.
\end{align}
$x_0$ is always present, while $x_\pm$ are stationary points only when the
square root is real, that is when 
\begin{align}
(p^2+q^2)^2-16p^2q^2>0.
\end{align}
Setting $z=p/q$ this means $x^4-14x^2+1>0$ so (since $p/q>1$) 
\begin{align}
x^2> 7+\sqrt {48}=(2+\sqrt 3)^2,
\end{align}
and, finally, 
\begin{align}
\frac pq> 2+\sqrt 3.
\end{align}
Taking the second derivative of the above expression and evaluating it in $%
x_0$, we get that $x_0$ is the absolute minimum (at fixed $p$ and $q$) if 
\begin{align}
9-\frac pq-\frac qp>0,
\end{align}
that is (recalling $p\geq q$), for 
\begin{align}
1\leq \frac pq < \frac 12 (9+\sqrt {77}),
\end{align}
otherwise the minimum is placed in $x_\pm$. In conclusion 
\begin{widetext}
\begin{align}
g_{a,min}(p,q)&=\frac {\pi}{3\sqrt 2} \left[2+\sqrt 5 \chi_a(p,q) \right], \\
\chi_a(p,q)&= 
\begin{cases}
1+10 \frac {p^2 q^2}{(p^2+q^2)^2}-6\frac {pq}{p^2+q^2}, & {\quad \mathrm{if}
\quad 1<\frac pq < \frac 12 (9+\sqrt {77})} \\ 
1-2 \frac {p^2 q^2}{(p^2+q^2)^2}, & {\quad \mathrm{otherwise}}.%
\end{cases}%
\end{align} 
\end{widetext}
The absolute minimum in the family is the minimum of the first row. Setting $%
x=pq/(p^2+q^2)$, we see that $1+10x^2-6x$ has a minimum for $x=3/10$, which
correspond to $p=3$, $q=1$. The corresponding absolute minimal energy is
exactly (\ref{minimal g}). This is not surprising at all, since the $%
(p,q)=(3,1)$, $\varepsilon=(1,1,1)$ corresponds to solution (\ref{gmin}) for 
$N=4$ (use (\ref{cj}) with $\Lambda=2$ in (\ref{lambdapm})). This
corresponds to the undeformed $SU(2)$ embedding, as anticipated.


\subsubsection{The case $\protect\varepsilon_b =(1,1,-1)$}

In this case we get 
\begin{align}
v_b=i\mathrm{diag}(1,0,-1,0), \qquad \|v_b\|^2=2.
\end{align}
Regarding the ranges, we can do de same exact reasoning as for the
previous case, so we get 
\begin{equation}
U_b(t,\varphi,r,\gamma)=g_4(\sigma_{p-q}\Phi) 
{\scriptstyle \begin{pmatrix}
e^{\frac 12 i r} & 0 & 0 & 0 \\ 
0 & 1 & 0 & 0 \\ 
0 & 0 & e^{-\frac 12 i r} & 0 \\ 
0 & 0 & 0 & 1%
\end{pmatrix}}
g_4(m\gamma),
\end{equation}
\begin{equation}
\Phi=\frac t{L_\varphi}-\varphi, \qquad \sigma_{p-q}= 
\begin{cases}
1 & \mbox{ for $p-q$ odd} \\ 
\frac 12 & \mbox{ for $p-q$ even}%
\end{cases}%
.
\end{equation}
\begin{widetext}
The Baryonic charge is 
\begin{align}
B=2\sigma_{p-q} m (p^2+q^2).
\end{align}
For the $g$-factor we get 
\begin{align}
g_b(p,q,m,& \tau)= L_rL_\gamma L_\phi \frac {K\pi^3}{4\sigma_{p-q} m} \left[ 
\frac {16\sigma^2_{p-q}}{L^2_\varphi} + \frac {2}{(p^2+q^2)L_r^2}+\frac {%
\sigma^2_{p-q}\lambda }{L^2_\varphi L^2_r} + 8\frac {m^2 }{L_\gamma^2}
\left( 1 + \frac {\lambda }{16 L_r^2} \right) \right. \cr + & \left. 8\frac {%
m^2 }{L_\gamma^2} \frac {\sigma^2_{p-q}\lambda(p^2+q^2)}{L^2_{\varphi} }
\left( 1+\frac {p^2q^2}{(p^2+q^2)^2}+\frac {3\tau^4}{(p^2+q^2)^2} -\frac {%
3\tau^2}{p^2+q^2} \right) \right]\!\!.
\end{align} 
The corresponding minimal energy, given by (\ref{energia minima}), in this
case becomes: 
\begin{align}
g_b(p,q,\tau)=\frac {\pi}{3\sqrt 2} \left[2+\sqrt 2 \left( 1+\frac {p^2q^2}{%
(p^2+q^2)^2}+\frac {3\tau^4}{(p^2+q^2)^2} -\frac {3\tau^2}{p^2+q^2}
\right)^{\frac 12} \right].
\end{align}
\end{widetext}
We can further minimise w.r.t. $\tau$. Setting $x=\tau^2$, it is immediate
to see that in this case the minimum is reached for 
\begin{align}
x_0=\frac {p^2+q^2}{2},
\end{align}
to which it corresponds the value 
\begin{align}
g_{b,min}(p,q)&=\frac {\pi}{3\sqrt 2} \left[2+\sqrt 2 \left( \frac 14+ 
\frac {p^2 q^2}{(p^2+q^2)^2} \right)^{\frac 12} \right].
\end{align}
For fixed $q$, this is a monotonic decreasing function of $p$, so there is
no an absolute minimum in this family. However, notice that the lower bound
is 
\begin{align}
g_{b,bound}&=\frac {\pi}{3\sqrt 2} \lim_{p\to\infty}\left[2+\sqrt 2 \left(
\frac 14+ \frac {p^2 q^2}{(p^2+q^2)^2} \right)^{\frac 12} \right]\cr &=\frac \pi6
(1+2\sqrt 2),
\end{align}
which is (\ref{minimal g}).\newline
We finally notice that this kind of solutions are not deformations of an $%
SU(2)$ or $SO(3)$ embedding, despite one could suspect it. Indeed, $v_b$ may
at most belong to the representation\footnote{%
We are using the convention that $\pmb s$ indicates the representation of
spin $s$} $\pmb {\frac 12}\oplus \pmb 0 \oplus \pmb 0$, or $\pmb 1 \oplus %
\pmb 0$ embedded in $SU(4)$. If so, there should exist a deformation of $%
k_{\underline c}$, t.i. a particular value of the moduli, such that $%
k_{\underline c}$ belongs into the same representation. But in both cases
the particular solution would be embedded in $SU(3)$ also and then it would
require $q=0$ or $p=0$.


\subsubsection{The case $\protect\varepsilon_c =(1,-1,1)$}

In this case we have 
\begin{align}
v_c=\frac i2(1,-1,1,-1), \qquad \|v_c\|^2=1.
\end{align}
Reasoning as before, we see that the field is now 
\begin{align}
U_c=&g_4(\sigma_{p-q}\Phi) 
\begin{pmatrix}
e^{\frac 14 i r} & 0 & 0 & 0 \\ 
0 & e^{-\frac 14 i r} & 0 & 0 \\ 
0 & 0 & e^{\frac 14 i r} & 0 \\ 
0 & 0 & 0 & e^{-\frac 14 i r}%
\end{pmatrix}
g_4(m\gamma), \\
\Phi=&\frac t{L_\varphi}-\varphi, \cr  \sigma_{p-q}= &
\begin{cases}
1 & \mbox{ for $p-q$ odd} \\ 
\frac 12 & \mbox{ for $p-q$ even}%
\end{cases}%
.
\end{align}
The Baryonic charge is 
\begin{align}
B=2\sigma_{p-q} m (p^2+q^2).
\end{align}
For the $g$-factor we get 
\begin{widetext}
 \begin{align}
g_c(p,q,m,& \rho)= L_rL_\gamma L_\phi \frac {K\pi^3}{\sigma_{p-q} m} \left[ 
\frac {16\sigma_{p-q}^2}{L^2_\varphi} + \frac {1}{2(p^2+q^2)L_r^2}+\frac {%
\lambda \sigma^2_{p-q}}{L^2_\varphi L^2_r} + 8\frac {m^2 }{L_\gamma^2}
\left( 1 + \frac {\lambda }{16 L_r^2} \right) \right. \cr + & \left.
8\sigma_{p-q}^2\frac {m^2 }{L_\gamma^2} \frac {\lambda(p^2+q^2)}{%
L^2_{\varphi} } \left( 1-\frac {2p^2q^2}{(p^2+q^2)^2} \right) \right]\!\!.
\end{align}
\end{widetext}
The corresponding minimal energy, given by (\ref{energia minima}), in this
case becomes: 
\begin{align}
g_c(p,q,|\rho|)=\frac {\pi}{3\sqrt 2} \left[2+\left( 1-2 \frac {p^2q^2}{%
(p^2+q^2)^2} \right)^{\frac 12} \right].
\end{align}
This is independent on $\tau$ and for fixed $q$ it is a monotonic increasing
function of $p$. It follows that the lower bound is reached for $p=q=1$ (the
value 1 is enforced by the request that $p$ and $q$ are coprime, but the
result depends only on $p/q$) 
\begin{align}
g_{c,bound}&=g_c(1,1)=\frac \pi6 (1+2\sqrt 2),
\end{align}
which, again, is (\ref{minimal g}). However, this is not allowed, since for $%
p=q=1$ the functions $f_j$ are not periodic and the solution of the
equations does not yield a well defined global solution! In this particular
family the absolute minimum is instead 
\begin{align}
g_{c,bound}&=g_c(2,1)=\frac \pi{3\sqrt 2} (2+\frac {\sqrt {17}}5)\simeq
2.0916,
\end{align}


\subsubsection{The case $\protect\varepsilon_d =(1,-1,-1)$}

In this case 
\begin{equation*}
v_{d}=i(0,-1,0,1),\qquad \Vert v_{d}\Vert ^{2}=2.
\end{equation*}%
This case looks to be very similar to the case $b$. Indeed, the reader can
be easily check that the matrices $v_{b}$, $k_{\underline{c}}$ transform
into $v_{d}$, $k_{\underline{c}}$ under the map 
\begin{align}
\mathrm{Mat}(N,\mathbb{C})& \longrightarrow \mathrm{Mat}(N,\mathbb{C}),\cr
a_{j,k} &\longmapsto a_{N-j,N-k}, \\
\mathbb{T}^{3}\times W& \longrightarrow \mathbb{T}^{3}\times W,\cr 
(e^{i\alpha_{1}},e^{i\alpha _{2}},e^{i\alpha _{3}},\tau )& \longmapsto (e^{i\alpha
_{3}},e^{i\alpha _{2}},e^{i\alpha _{1}},pq/\tau ). \cr
\end{align}%
Under this map the inverse Cartan matrix is invariant and $\underline{%
\varepsilon }_{b}\mapsto -\underline{\varepsilon }_{d}\equiv \underline{%
\varepsilon }_{d}$, where the last equivalence is by a global rescaling.
This sort of duality makes the two families perfectly equivalent and giving
the same minima.\newline
\textbf{Remark:} We see that of the four predicted sequences of families the
true inequivalent ones are the first three ones, while the $d$ case is not
really new. It is natural to expect that such duality extends to any $N$,
but this would require a deeper understanding of the global properties of
the relevant moduli space $W$. To this aim, it would be interesting to
investigate the explicit cases $N=5$ and $N=6$. This, however, goes beyond
the scope of the present work.


\section{Shear modulus for lasagna states}

On the
crust of ultra compact objects, like neutron stars, nucleons form large structures called
pasta states. Knowing the elasticity properties of the crust may be very
important to understand the structure of the gravitational waves emitted in
a collision with a black hole.
An important recent result has been found in \cite{Caplan:2018gkr}
where, using numerical simulations based on the phenomenological
nucleon-nucleon potential, the authors showed that the shear modulus for
nuclear lasagna can have a value much larger than previous estimates. Here
we give a first principle explanation of it as an application of the Skyrmionic model.  To compute the shear modulus
associated to lasagna states, our strategy will be to first compute it for
the $SU(2)$ case for the solutions determined in \cite{Alvarez:2017cjm, Canfora:2019asc}, by
employing its relation with the $1+1$ computations presented in \cite%
{Takayama:1992eu}.

Let us begin with a review \cite{Canfora:2019asc}. We will consider the
symmetric case\footnote{%
Notice that we are referring to the $p$ and $q$ in \cite{Canfora:2019asc}, which
have different meaning than the $p$ and $q$ used in the previous section} in
Equations (13) and (16) of \cite{Canfora:2019asc}, namely%
\begin{equation*}
p=q\ ,\ l_{2}=l_{3}=\sqrt{A}\ .
\end{equation*}%
This means that we are considering configurations in which the $SU(2)$
Skyrmions live in a box of volume $V_{tot}$,%
\begin{equation*}
V_{tot}=16\pi ^{3}Al_{1}
\end{equation*}%
where $l_{1}$ is the length along the $r$ direction (which is the coordinate
of the profile $H$ in Eq. (13) of \cite{Canfora:2019asc}). The Baryonic charge
corresponding to the Ansatz in Equations (12), (13) and (14) of \cite%
{Canfora:2019asc} is%
\begin{equation}
B=pq=p^{2}  \label{rescaled0}
\end{equation}%
(see below Eq. (24) page 5 of \cite{Canfora:2019asc}). Then, the SU(2) field
equations for the Ansatz in Equations (12), (13), (14) and (16) of \cite%
{Canfora:2019asc} with a static profile $H=H(r)$, reduce to%
\begin{equation}
-\frac{d^{2}u}{dr^{2}}+\Gamma ^{2}\sin u=0\ ,  \label{rescaled1}
\end{equation}%
where%
\begin{eqnarray}
u(r) &=&4H(r)\ ,\ \ \ 0\leq r\leq 2\pi \ ,  \label{rescaled2} \\
\Gamma ^{2} &=&\left( \frac{B}{A}\right) ^{2}\frac{\lambda l_{1}^{2}}{%
4+2\lambda \frac{B}{A}}\ ,  \label{rescaled3}
\end{eqnarray}%
where $\frac{B}{A}$ can be interpreted as the Baryon density per unit of
area of the Lasagna configuration (up to $\pi $ factors). In order to
compare directly the present results with the ones in \cite{Takayama:1992eu}, it
is convenient to define the rescaled coordinate $y$ as follows%
\begin{equation}
y=\Gamma r\ ,\ \ 0\leq y\leq 2\pi \Gamma \ ,  \label{rescaled4}
\end{equation}%
so that the field equation (\ref{rescaled1}) becomes%
\begin{equation}
-\frac{d^{2}u}{dy^{2}}+\sin u=0\Leftrightarrow \frac{\left( \frac{du}{dy}%
\right) ^{2}}{2}=1-\cos u+C\ ,  \label{rescaled5}
\end{equation}%
and the boundary conditions in order to have Baryonic charge $B=pq=p^{2}$ are%
\begin{equation}
H(2\pi )=\frac{\pi }{2}\Leftrightarrow u\left( 2\pi \Gamma \right) =2\pi \ .
\label{rescaled6}
\end{equation}%
Now, equations (\ref{rescaled4}), (\ref{rescaled5}) and (\ref{rescaled6})
(which are equivalent to the results in \cite{Canfora:2019asc}) can be compared
directly with equations (2.4), (2.7) and (2.9) of \cite{Takayama:1992eu}. In
particular, the dictionary between the results of \cite{Takayama:1992eu} and the
present ones is%
\begin{eqnarray}
\phi \left( x\right)  &\rightarrow &u\left( y\right) \ ,  \label{dictionary1}
\\
L &\rightarrow &2\pi \Gamma \ ,  \label{dictionary2} \\
k &\rightarrow &\sqrt{\frac{2}{C+2}}\equiv \tau \ ,\   \label{dictionary2.5}
\\
k^{\prime } &\rightarrow &\sqrt{\frac{C}{C+2}}  \label{dictionary2.6}
\end{eqnarray}%
where the left hand side (with respect to the \textquotedblleft $\rightarrow 
$\textquotedblright ) is from \cite{Takayama:1992eu} while the right hand side
comes from the above equations. Equations (2.9) and (2.10) of \cite%
{Takayama:1992eu} read%
\begin{eqnarray*}
L &=&2\tau I_{-1/2}\left( \tau \right) \ , \\
I_{-1/2}\left( \tau \right)  &=&\int_{0}^{\pi /2}dy\left( 1-\tau ^{2}\sin
^{2}y\right) ^{-1/2}\ ,
\end{eqnarray*}%
that is 
\begin{equation}
\pi \Gamma =\sqrt{\frac{2}{C+2}}\int_{0}^{\pi /2}dy\left( 1-\frac{2}{C+2}%
\sin ^{2}y\right) ^{-1/2}\ , \label{dictionary3} 
\end{equation}%
which fixes the integration constant $C$ in Eq. (\ref{rescaled5}) in terms
of $\Gamma $ 
\begin{equation*}
C=C\left( \Gamma \right) \ ,
\end{equation*}%
which depends on the Baryon charge as well as the size of the box in which
the configuration lives. Now, with the above dictionary, we can write the
speed of sound of the phonons using Eq. (3.15) of \cite{Takayama:1992eu}:%
\begin{align*}
V_{phonons}=&\sqrt{\frac{C}{2}}\frac{\pi \Gamma }{\int_{0}^{\pi /2}dy\left( 1-%
\frac{2}{C+2}\sin ^{2}y\right) ^{1/2}}\cr =&\sqrt{\frac{G_{SU(2)}}{T_{00}}},
\end{align*}%
where the $T_{00}$ is given in Eq. (28) of \cite{Canfora:2019asc}. Thus we have the
following expression for the shear modulus $G_{SU(2)}$ in the $SU(2)$ case%
\begin{equation*}
G_{SU(2)}=\left( V_{phonons}\right) ^{2}T_{00}\ .
\end{equation*}%
We can then estimate it as follows. In place of $T_{00}$ we use its mean
value computed as 
\begin{equation*}
\bar{T}_{00}=\frac{E_{min}^{0}B}{16\pi ^{3}l_{1}A},
\end{equation*}%
where $E_{min}^{0}$ is the minimal energy corresponding to $B=1$. From Table
1 of  \cite{Takayama:1992eu} we see that $B/A$ is independent from $B$ for the
minimal energy configuration. Using the values in the table\footnote{%
Notice that with these values the baryon density is $n\simeq
0.0468fm^{-3}\approx 0.05fm^{-3}$, the same value used in the simulations of 
\cite{Caplan:2018gkr}} we get 
\begin{equation*}
\bar{T}_{00}\simeq 1.26\ 10^{34}erg/cm^{3}.
\end{equation*}%
With the same values, from (\ref{rescaled3}) we obtain 
\begin{equation*}
\Gamma \simeq 0.371,\quad \ \pi \Gamma \simeq 1.166.
\end{equation*}%
Therefore condition (\ref{dictionary3}), which is easily solved numerically
after noticing that $I_{-\frac{1}{2}}(\tau )=K(\tau ^{2})$, the first
complete elliptic integral, gives 
\begin{equation*}
C\simeq 2.73
\end{equation*}%
and 
\begin{equation*}
V_{phonons}\simeq 0.1198.
\end{equation*}%
Finally, 
\begin{equation*}
G_{SU(2)}\simeq 1.8\ 10^{32}erg/cm^{3}.
\end{equation*}%
Notice that the present value is expected to an approximation from above,
since we are using a Skyrmionic effective model. From the above analysis,
taking into account (\ref{minimal g}), we can infer that in any case the
true value should be $G_{SU(2)}\gtrsim 10^{31}erg/cm^{3}$. The comparison with \cite%
{Caplan:2018gkr} is very good especially taking into account that we only used the
Skyrme model.
\newline
At this point we can use the new solutions found in the present work to
relate the shear modulus for $SU(N)$ case to the one for $SU(2)$.\newline
Let us consider the minimal energy per nucleon (\ref{g-min-dim}). After
multiplying by $B$ and dividing by the volume, which, because of (\ref{lati}%
) is proportional to $\lambda ^{\frac{3}{2}}$, we get 
\begin{equation*}
{\bar{T}}_{00}\propto \frac{K}{\lambda }N(N^{2}-1).
\end{equation*}%
On the other hand, the Baryon density is 
\begin{equation*}
n=\frac{B}{8\pi ^{2}L_{\varphi }L_{r}L_{\gamma }}\propto \frac{N(N^{2}-1)}{%
\lambda ^{3/2}},
\end{equation*}%
which solved for $\lambda $ and replaced in ${\bar{T}}_{00}$ gives 
\begin{equation*}
{\bar{T}}_{00}\propto n^{2/3}\sqrt[3]{N(N^{2}-1)}.
\end{equation*}%
Assuming the speed of sound to be essentially independent from $N$, as
suggested by the fact that all the component of $T_{\mu \nu }$ scale in the
same way with $N$, we get that the dependence of the shear modulus from $N$
is 
\begin{equation*}
G_{SU(N)}\propto \sqrt[3]{N(N^{2}-1)},
\end{equation*}%
so that we get the final estimate for the value of the shear modulus $%
G_{SU(N)}$ of the $SU(N)$ Skyrme model as 
\begin{eqnarray*}
G_{SU(N)} &=&a(N)G_{SU(2)}\ , \\
a(N) &=&\sqrt[3]{\frac{N(N^{2}-1)}{6}}\ .
\end{eqnarray*}


\section{Conclusion and perspectives}

\textit{In conclusion}, we constructed the first examples of analytic
(3+1)-dimensional Skyrmions living at finite Baryon density in the SU(N)
Skyrme model (which are not trivial embeddings of $SU(2)$ into $SU(N)$) for
any N. These results allow to compute explicitly the energy to Baryon charge
ratio for any N and to discuss its smooth large N limit as well as the
closeness to the BPS bound. The energy density profiles of these finite
density Skyrmions have \textit{lasagna-like} shape.  A quite
remarkable by-product of the present analysis is that we have been able to
estimate analytically 
the shear modulus of lasagna-shaped configurations
which appear at finite Baryon density. Our estimate agrees with recent
results \cite{Caplan:2018gkr} based on many body simulations in nuclear physics
using phenomenological nucleon-nucleon interaction potentials.


\begin{acknowledgments}
This work has been funded by Fondecyt Grant 1200022, MINEDUC-UA project ANT 1755 and by Semillero de Investigaci\'on SEM 18-02 of the Universidad de Antogasta. Centro de Estudios Cient\'{\i}ficos (CECs) is funded by the Chilean Government through the Centers of Excellence Base Financing Program of Conicyt
\end{acknowledgments}

\appendix
\section{General facts and conventions about $SU(N)$}\label{app:su(n)}
In this section we collect some general facts we applied for getting the solutions. Let $V, (|)$ the $N$-dimensional complex vector space isomorphic to $\mathbb C^N$, endowed with the canonical hermitian product
\begin{align*}
 & (\underline z| \underline w )=\sum_{j=1}^n z^*_j w_j, \qquad z,w\in \mathbb C^n, \cr
 & (x+iy)^*=x-iy, \quad x,y\in \mathbb R. 
\end{align*}
The unitary group $U(N)\equiv U(V)$ is the group of unitary transformations of $V$. Looking $U(V)$ as automorphisms of $V$ determines the smallest fundamental representation, simply called $V$. The action of $U(V)$ over $V$ induces an action on the
external products $\wedge^k V$ of $V$, and the corresponding homomorphisms 
\begin{align*}
 U(V) \longrightarrow {\rm Aut}(\wedge^k V), \quad, k=1,\ldots, N
\end{align*}
are all representations, also called  $\wedge^k V$. For $k=1,\ldots, N-1$ are all faithful (t.i. the kernel of the map is the identity transformation) and are called {\it the fundamental representations}. Any other finite dimensional representation is obtained by
their tensor products. $\wedge^N V$ is not faithful. The corresponding kernel is a normal subgroup of $U(N)$ called the special group $SU(N)\equiv SU(V)$. \\
$SU(N)$ is a compact simply connected simple Lie group of rank $N-1$. It essentially means that it contains a maximal abelian torus $T_N$ of dimension $N-1$. On $V$, it is represented by the diagonal $N\times N$ matrices $T$ such that
\begin{align*}
 \prod_{j=1}^N T_{jj}=1, \qquad  |T_{jj}|=1, \quad j=1,\ldots,N. 
\end{align*}
The center $Z_N$ of $SU(N)$ is the subgroup of $T$ consisting of the elements commuting with the whole $SU(N)$ (equivalently, it is the kernel of the adjoint representation). It consists of matrices of the form $\omega \mathbb I$, where 
$\omega^N=1$ and $\mathbb I$ is the identity matrix. Therefore, $Z_N\simeq \mathbb Z_{N}$. All the other compact simple Lie groups locally isomorphic to $SU(N)$ are the quotients $SU(N)_{\Gamma}:=SU(N)/\Gamma$, where $\Gamma$ is any given 
subgroup of $Z_N$. They are not simply connected, since their first homotopy group is $\pi_1(SU(N)_{\Gamma})=Z_N/\Gamma$. $SU(N)$ is the universal covering of all of them. In particular, for $N=2$ we have just two groups, which are $SU(2)$ and 
$SU(2)_{\mathbb Z_2} \simeq SO(3)$.

To any Lie group $G$ one associates the corresponding Lie algebra $L(G)$, which is the algebra of left invariant vector 
fields\footnote{t.i. the vector fields invariant under the action of the left translation $L_g: G\rightarrow G$, $L_g(h)=gh$, for any given $g\in G$} 
over $G$, endowed with the Lie bracket product. In matrix representation it reduces to the commutator $[,]$. Since the groups $SU(N)_\Gamma$ are locally isomorphic to $SU(N)$, their Lie algebras are all isomorphic. One gets
\begin{align}
\mathfrak{su}(N)&\equiv Lie(SU(N))\cr
=&\{ X\in Mat(N)| X^\dagger =-X,\  {\rm Tr} X=0 \},
\end{align}
t.i. the antihermitian traceless $N\times N$ complex matrices.\\
In particular, $H:=Lie(T_N)$ is a maximal abelian subalgebra of $\mathfrak{su}(N)$, having the property that, for any $X\in H$, the linear map $ad_X : \mathfrak{su}(N)\rightarrow \mathfrak{su}(N)$ defined 
by\footnote{This is called the adjoint action and defines the adjoint representation of the algebra over itself} $ad_X(Y)=[X,Y]$ for any $Y\in \mathfrak{su}(N)$, is diagonalizable (on the complexification
of the algebra).\\
We see from the definition that $\mathfrak{su}(N)$ is a real vector space of dimension $N^2-1$. A basis can be easily obtained as follows. For any $j,k=1,\ldots,N$ we define the matrix $E_{j,k}$ with elements 
\begin{align}
(E_{j,k})_{mn}=\delta_{jm} \delta_{kn}.
\end{align}
They are called the {\it elementary matrices}. With these notations, a basis of $\mathfrak{su}(N)$ is given by
\begin{align}
A_{j,k}&= (E_{j,k}-E_{k,j}), \quad S_{j,k}=i(E_{j,k}+E_{k,j}), \cr & 1\leq j<k \leq N, \\
J_h&=i(E_{j,j}-E_{j+1,j+1}), \quad \ h=1,\ldots,N-1.
\end{align} 
In particular, the matrices $J_h$ span the Cartan subalgebra $H$.
\subsection{Roots and simple roots}\label{sec: simple roots}
A concept that is particularly helpful for most of the calculations we need is the one of roots. These are related to the above observation regarding the diagonalizability of $ad_X$ for any $X\in H$.
The diagonalizability must be checked on $\mathfrak{su}(N)\otimes \mathbb C$, which is generated by the complex span of the basis given above, in place of the real span. Notice that the complex span contains the matrices $E_{i,j}$, $i\neq j$. This is sufficient
to determine all the eigenvectors and eigenvalues of $ad_X$ for any given $X\in H$.\\
To this aim, let us specify $H$ as follows:
\begin{align}
H=\left\langle X=i\sum_{j=1}^{n} c_j E_{j,j} | \sum_i c_i=0 \right\rangle_{\mathbb{R}},
\end{align}
where with $\langle \cdots \rangle_{\mathbb R}$ we mean the span over $\mathbb R$ of $\cdots$.
Thus, we immediately see that
\begin{align}
[X,E_{j,k}]&=i(c_j-c_k)E_{j,k},\\
[X,J_h]&=0,
\end{align}
so that $E_{j,k}$ and $J_h$ are eigenmatrices of the adjoint action of $X$, with eigenvalues $i(c_j-c_k)$ and $0$ respectively. The point is that the eigenvalues depend linearly on $X$. Let us consider the linear operators $L_j$, $j=1,\ldots,N$ defined
by
\begin{align*}
L_j: Mat(N) \longrightarrow \mathbb C, \ A\longmapsto A_{j,j}.
\end{align*}
Then, we can write $ic_j=L_j(X)$ so that 
\begin{align*}
 ad_X(E_{j,k}) =(L_j-L_k)(X) E_{j,k}.
\end{align*}
The linear operators
\begin{align}
\beta_{i,j} :=L_j-L_k : H\longrightarrow \mathbb C
\end{align}
are said the non vanishing {\it roots} of $\mathfrak{su}(N)$. The corresponding eigenspaces are one dimensional. Beyond these, there is the vanishing root defining the $0$ eigenvalue, which eigenspace is $H$, so that has degeneration equal to the
rank $r=N-1$.\\
In particular, the set of non vanishing roots contains a set of $r$ linearly independent roots, having the property that all the remaining roots are are combination of them with all non positive or all non negative integer coefficients. These are called
the {\it simple roots} and are  
\begin{align}
 \beta_j:=L_j -L_{j+1},\quad\ j=1,\ldots,N-1.
\end{align}
Finally, for convenience, we introduce the less conventional concept of {\it real valued roots} $\alpha_{j,k}=-i\beta_{j,k}$, $\alpha_j=-i\beta_j$, which we will simply call again roots and simple roots. With this convention, for the simple roots $\alpha_j$, 
we can also write
\begin{align}
 \alpha_j: H\longrightarrow \mathbb R, \quad X\longmapsto -{\rm Tr} (J_j X), \label{tralphaj}
\end{align}
useful for practical purposes. This also shows that the $\alpha_j$ are linearly independent. 
We name the corresponding eigenvectors $\lambda_{\alpha_j}\equiv \lambda_j=E_{j,j+1}$, so that
\begin{align}
[X,\lambda_j]=i\alpha_j(X)\lambda_j , \qquad \forall X\in H.  \label{simple roots}
\end{align}
Notice that $\lambda_{-\alpha_j}=\lambda_{\alpha_j}^\dagger$, so
\begin{align}
[X,\lambda^\dagger_j]=-i\alpha(X)\lambda^\dagger_j . 
\end{align}
\subsection{Some further technical facts}\label{app:sftf}
There is a canonical way to introduce a scalar product on the real space spanned by the simple roots. We however bypass the historical construction and employ (\ref{tralphaj}) to define the scalar product
\begin{align}
(\alpha_j | \alpha_k) := -{\rm Tr} (J_j J_k). 
\end{align}
On 
\begin{align*}
H^*_{\mathbb R}:= \langle \alpha_1,\ldots,\alpha_{N-1} \rangle_{\mathbb R} 
\end{align*}
it is an euclidean scalar product. One then defines the $r\times r$ {\it Cartan matrix}\footnote{The name comes from the fact that in the Dynkin classification the algebra $su(N)$ is called $A_r$, where $r$ is the rank} $C_{A_{N-1}}$ with components
\begin{align*}
(C_{A_{N-1}})_{j,k}:=& 2 \frac {( \alpha_j | \alpha_k)}{(\alpha_j|\alpha_j)}=( \alpha_j | \alpha_k) \cr
=& 2 \delta_{j,k}-\delta_{j,k+1}-\delta_{j,k-1}.  
\end{align*}
The Cartan matrix is strictly positive definite. Indeed, for any vector $(x^1,\ldots, x^r)\in \mathbb R^r$ we have
\begin{align}
\sum_{j,k}&  x^jx^k  (C_{A_{N-1}})_{j,k}= 2\sum_{j=1}^{r}  x_j^2 - \sum_{j=1}^{r-1} 2 x_j x_{j+1} \cr
&=x_1^2+x_r^2 +\sum_{j=1}^{r-1} (x_j-x_{j+1})^2,
\end{align}
which is strictly positive and vanishes only for $x_j=0$ for all $j$.
In particular, the Cartan matrix is invertible and, indeed, one easily checks that
\begin{align}
(C_{A_{N-1}}^{-1})_{j,k} =\frac 1N {\rm min}(j,k) (N-{\rm max}(j,k)). \label{Cartinversa}
\end{align}
Another important fact to notice is that for $j,k$ one has
\begin{align}
 [\lambda_j,\lambda_k]= \delta_{j+1,k} E_{j,j+2}. \label{lambdajlambdak}
\end{align}

\section{Proof of Proposition \ref{prop:eqmotion}}\label{app:prop1}
In order to prove the proposition, it is convenient to work with the coordinates $\Phi$ and $T=t+L_\varphi \varphi$. The metric takes the form $ds^2=-L_\varphi dT d\Phi+L_r^2 dr^2+L_\gamma^2 d\gamma^2$. With the given Ansatz, after replacing
$\Phi$ with $\sigma \Phi$ for constant $\sigma$ (for convenience), for the $L_\mu$ we get
\begin{align*}
 R_T&=0, & R^T&=-\frac 2{L_\varphi} R_\Phi, \\
 R_\Phi&=\sigma e^{-m\gamma k} e^{-h(r)} k e^{h(r)} e^{m\gamma k},  & R^\Phi &=0, \\
 R_r&=e^{-m\gamma k} h'(r) e^{m\gamma k}, & R^r&=\frac 1{L_r^2} R_r, \\
 R_\gamma &= mk, & R^\gamma&=\frac 1{L_\gamma^2} R_\gamma.
\end{align*}
For $F_{\mu\nu}=[R_\mu,R_\nu]$, with $x=e^{-h(r)}ke^{h(r)}$, we get the non vanishing components
\begin{align*}
 & F_{\Phi r}=-F_{r\Phi}=\sigma e^{-m\gamma k}[x,h'] e^{m\gamma k}, \\
 & F_{\Phi \gamma}=-F_{\gamma\Phi}=\sigma m e^{-m\gamma k}[x,k] e^{m\gamma k}, \\
 & F_{r\gamma}=-F_{\gamma r}=me^{-m\gamma k}[h',k] e^{m\gamma k}.
\end{align*}
Setting $\mathcal L_\mu:=[L^\nu,F_{\mu\nu}]$, the equations of motion are
\begin{align}
 0=\partial^\mu R_\mu +\frac \lambda4 \partial^\mu \mathcal L_\mu. \label{equazione del moto}
\end{align}
Using that nothing depends on $T$ and that there are no lower $T$ components, these reduce to
\begin{align*}
 0=\frac 1{L_r^2} \partial_r (R_r+\frac \lambda4 \mathcal L_r)+\frac 1{L_\gamma^2} \partial_\gamma (R_\gamma+\frac \lambda4 \mathcal L_\gamma).
\end{align*}
But
\begin{align*}
 \partial_\gamma R_\gamma&=0, \\
 \partial_r R_r&=e^{-m\gamma k} h'' e^{m\gamma k}, \\
 \partial_r \mathcal L_r&=\partial_r \left( \frac {m^2}{L_\gamma^2} e^{-m\gamma k} [k,[h',k]] e^{m\gamma k} \right) \cr
 &=\frac {m^2}{L_\gamma^2} e^{-m\gamma k} [k,[h'',k]] e^{m\gamma k} , \\
 \partial_\gamma \mathcal L_\gamma&=-\frac {m}{L_r^2}\partial_\gamma \left(  e^{-m\gamma k} [h',[h',k]] e^{m\gamma k} \right)\cr
& =\frac {m^2}{L_r^2} e^{-m\gamma k} [k,[h',[h',k]]] e^{m\gamma k} ,
\end{align*}
so (\ref{equazione del moto}) becomes
\begin{widetext}
 \begin{align*}
 0=\frac 1{L_r^2} e^{-m\gamma k} \left( h''-\frac \lambda4 \frac {m^2}{L_\gamma^2} \left( [k,[k,h'']]-[k,[h',[h',k]]] \right) \right) e^{m\gamma k},
\end{align*}
which proves the proposition.

\subsection{Further details}\label{app:further details}
Making use of (\ref{kappa}) and (\ref{simple roots}) we can write
\begin{align}
 [h',k]&=\sum_{j=1}^{N-1} (c_j[h',\lambda_j]-c^*_j [h',\lambda_j^\dagger])=i\sum_{j=1}^{N-1}\alpha_j(h') (c_j \lambda_j+c^*_j \lambda_j^\dagger).\label{hprimok}
\end{align}
Repeating the same calculation:
\begin{align*}
[h',[h',k]]&=i\sum_{j=1}^{N-1} (\alpha_j(h') c_j[h',\lambda_j]+\alpha_j(h') c^*_j [h',\lambda_j^\dagger])=-\sum_{j=1}^{N-1}\alpha_j(h')^2 (c_j \lambda_j-c^*_j \lambda_j^\dagger).
\end{align*}
Finally, 
 \begin{align*}
[k,[h',[h',k]]]=&-\sum_{k=1}^{N-1} \sum_{j=1}^{N-1} \alpha_j(h')^2 \left\{ c_kc_j [\lambda_k,\lambda_j] +c^*_k c^*_j [\lambda^\dagger_k,\lambda^\dagger_j] -c_kc^*_j [\lambda_k,\lambda^\dagger_j] -c^*_k c_j [\lambda^\dagger_k, \lambda_j] \right\}\\
=&-\sum_{k=1}^{N-1} \sum_{j=1}^{N-1} \alpha_j(h')^2 \left\{ c_kc_j [\lambda_k,\lambda_j] +c^*_k c^*_j [\lambda^\dagger_k,\lambda^\dagger_j] -c_kc^*_j [\lambda_k,\lambda^\dagger_j] +c^*_k c_j [\lambda_j, \lambda^\dagger_k] \right\}.
\end{align*}
The last two terms cancel after summation, while the first terms vanish for $j=k$, so we get
 \begin{align*}
[k,[h',[h',k]]]=&-\sum_{j<k} \alpha_j(h')^2 \left(c_j c_k [\lambda_k,\lambda_j]+c^*_jc^*k[\lambda^\dagger_k,\lambda^\dagger_j]\right)-\sum_{k<j} \alpha_j(h')^2 \left(c_j c_k [\lambda_k,\lambda_j]+c^*_jc^*k[\lambda^\dagger_k,\lambda^\dagger_j]\right)\\
=&\sum_{j<k} (\alpha_j(h')^2-\alpha_k(h')^2) \left(c_j c_k [\lambda_j,\lambda_k]+c^*_jc^*k[\lambda^\dagger_j,\lambda^\dagger_k]\right),
\end{align*}
where we have changed the order of commutators in the first sum and exchanged the name of variable in the second sum. Therefore
\begin{align}
[k,[h',[h',k]]]=&\sum_{j<k} \left(\alpha_j(h')^2-\alpha_k(h')^2\right)\left( c_jc_k [\lambda_j,\lambda_k] +c^*_j c^*_k[\lambda^\dagger_j,\lambda^\dagger_k] \right)\ . \label{khhk}
\end{align}
Similarly,
\begin{align*}
 [k,h'']=-[h'',k]=-i\sum_{j=1}^{N-1} \alpha_j(h'') (c_j \lambda_j+c^*_j \lambda^\dagger_j),
\end{align*}
and 
\begin{align*}
 [k,[k,h'']]=-i\sum_{j=1}^{N-1}\sum_{k-1}^{N-1} \alpha_j(h'') \left[ c_jc_k[\lambda_k,\lambda_j] -c_j^*c_k^* [\lambda_k^\dagger,\lambda_j^\dagger] -c_jc_k^* [\lambda_k^\dagger,\lambda_j]+c_j^* c_k[\lambda_k,\lambda^\dagger_j] \right].
\end{align*}
\end{widetext}
The first two terms can be treated as above, giving the contribution
\begin{align*}
i\sum_{j<k} (\alpha_j(h'')-\alpha_k(h''))\left( c_jc_k [\lambda_j,\lambda_k] -  c^*_jc^*_k [\lambda^\dagger_j,\lambda^\dagger_k]  \right) ,
\end{align*}
while the last two terms, after renaming the labels in the first of the sums, give the contribution
\begin{align*}
 -i \sum_{j=1}^{N-1}\sum_{k-1}^{N-1} \left(\alpha_j(h'')+\alpha_k(h'')\right)c_kc_j^* [\lambda_k,\lambda_j^\dagger].
\end{align*}
Now,
\begin{align*}
 [\lambda_k,\lambda^\dagger_j]=[E_{k,k+1}, E_{j+1,j}],
\end{align*}
which in components is
\begin{align*}
[E_{k,k+1}, E_{j+1,j}]_{m,r}
&=\delta_{j,k}(E_{j,j}-E_{j+1,j+1})_{m,r}
\end{align*}
so that
\begin{align}
 [\lambda_k,\lambda^\dagger_j]=-i\delta_{j,k} J_j.\label{commutatore}
\end{align}
We finally get 
\begin{eqnarray}
&&[k,[k,h'']]= i\sum_{j<k} \left(\alpha_j(h'')-\alpha_k(h'')\right)\cdot \label{kkh} \\
&&\cdot \left( c_jc_k [\lambda_j,\lambda_k] -
c^*_j c^*_k[\lambda^\dagger_j,\lambda^\dagger_k]  \right)-2 \sum_{j=1}^{N-1} \alpha_j(h'')|c_j|^2 J_j\ . \nonumber
\end{eqnarray} 

\subsection{A further proposition}
We want now to state another technical proposition:
\begin{prop}
 Assume $k_{\underline c}=\sum_{j=1}^{N-1} (c_j E_{j,j+1}-c_j^* E_{j+1,j})$, $h'\in H$ a matrix such that $\alpha_j(h')=:\varepsilon_j a$ where $\varepsilon_j$ is a sign, $j=1,\ldots,N-1$, and $x:=e^{-h' r} k_{\underline c}e^{h' r}$. Then
 \begin{align}
{\rm Tr} k^2_{\underline c}&=-2\|\underline c\|^2, \label{trakquadro}
\end{align}
\begin{align}
{\rm Tr} ([h',k_{\underline c}][h',k_{\underline c}])&={\rm Tr} ([h',x][h',x])=-2a^2 \| \underline c \|^2, \label{trhkhk}
\end{align}
and
\begin{eqnarray}
&&{\rm Tr} ([x,k_{\underline c}][x,k_{\underline c}])=-8\sin^2 (ar)  \label{trxkxk}\\
&&\left(\sum_{j=1}^{N-1} |c_j|^4+\sum_{j=1}^{N-2} |c_j|^2 |c_{j+1}|^2 \frac 12(1-3\varepsilon_j \varepsilon_{j+1}) \right). 
\nonumber
\end{eqnarray}   
\end{prop}
\begin{proof}
First, we have 
\begin{align}
 {\rm Tr}k^2_{\underline c}=\sum_{j=1}^{N-1} \sum_{k=1}^{N-1} \{ c_jc_k{\rm Tr} (\lambda_j \lambda_k) +c^*_jc^*_k{\rm Tr} (\lambda^\dagger_j \lambda^\dagger_k)\nonumber\\-
 c^*_jc_k{\rm Tr} (\lambda^\dagger_j \lambda_k) -c_jc^*_k{\rm Tr} (\lambda_j \lambda^\dagger_k) \},
\end{align}  
where we used the notation $\lambda_j=E_{j,j+1}$. Since $\lambda_j$ is upper diagonal, so is $\lambda_j \lambda_k$, hence ${\rm Tr}(\lambda_j \lambda_k)=0$. Similarly, ${\rm Tr}(\lambda^\dagger_j \lambda^\dagger_k)=0$.
On the other hand \newpage
\begin{align}
 {\rm Tr} (\lambda^\dagger_j \lambda_k)&=\sum_{n=1}^N \sum_{m=1}^N (E_{j+1,j})_{nm} (E_{k,k+1})_{mn}\cr &=\sum_{n=1}^N \sum_{m=1}^N \delta_{j+1,n} \delta_{jm}\delta_{k+1,n} \delta_{km}\cr &=\delta_{kj}= {\rm Tr} (\lambda_j\lambda^\dagger_k ).
\end{align}
This proves (\ref{trakquadro}).\\
Now, notice that 
\begin{align}
 [h',x]=[h',e^{-h'r}k_{\underline c} e^{h'r}]=e^{-h'r} [h',k_{\underline c}]e^{h'r}
\end{align}
since $h'$ commutes with $e^{h'r}$. Therefore,
\begin{align}
 {\rm Tr} ([h',x][h',x])&={\rm Tr} (e^{-h'r}[h',k_{\underline c}][h',k_{\underline c}]e^{h'r})\cr &={\rm Tr} ([h',k_{\underline c}][h',k_{\underline c}]),
\end{align}
because of the cyclicity property of the trace. So we are left with the computation of ${\rm Tr} ([h',k_{\underline c}][h',k_{\underline c}])$.
Using (\ref{hprimok}) and the fact that the only non vanishing traces are ${\rm Tr}(\lambda_j^\dagger \lambda_k)=\delta_{j,k}$, we get
\begin{widetext}
 \begin{align}
 {\rm Tr} ([h',k_{\underline c}][h',k_{\underline c}])&={\rm Tr} \left( \sum_{j=1}^{N-1} (i\alpha_j(h')\lambda_j c_j+i \alpha_j(h') c^*_j \lambda^\dagger_j)\sum_{k=1}^{N-1} (i\alpha_k(h')\lambda_k c_k+i \alpha_k(h') c^*_k \lambda^\dagger_k)  \right)\cr
 &=-2\sum_{j=1}^{N-1} \alpha_j (h')^2 c_jc^*_j =-2a^2 \|\underline c\|^2,
\end{align}
\end{widetext}
where we used that $ \alpha_j (h')^2=(\varepsilon_j a)^2=a^2$. This proves (\ref{trhkhk}).\\
Let us write $h=h'r$. Therefore,
\begin{align}
 x&=e^{-h} k_{\underline c} e^h=\sum_{j=1}^{N-1} (c_j e^{-h} \lambda_j e^{h} -h.c.).
\end{align}
Using the notation $ad_X(Y)=[X,Y]$ for any pair of matrices $X,Y\in \mathfrak{su}(N)$, we first notice the identity
\begin{align}
 e^{tX} Y e^{-tX}=\sum_{n=0}^\infty \frac 1{n!} t^n ad_X^n (Y), \label{XAX}
\end{align}
where with $ad_X^n$ we mean the iterated application of $ad_X$. Indeed,
\begin{align}
 \frac d{dt} (e^{tX} Y e^{-tX}) 
 &=e^{tX} ad_X(Y) e^{-tX}.
\end{align}

Hence
\begin{align}
 \left. \frac {d^n}{dt^n} \right|_{t=0} (e^{tX} Y e^{-tX})&=e^{tX} ad^n_X(Y) e^{-tX}|_{t=0}\cr &=ad^n_X(Y),
\end{align}
so that (\ref{XAX}) is the Taylor expansion of $e^{tX} Y e^{-tX}$. For $Y=k_{\underline c}$, $X=h$ and $t=-1$, and using that $ad_h(\lambda_j)=i\alpha_j(h)=i\varepsilon_j at$ we then have
\begin{align}
 e^{-h} \lambda_j e^{h}&=\sum_{n=0}^\infty \frac 1{n!} (-1)^n ad_h^n (\lambda_j)\cr &=\sum_{n=0}^\infty \frac 1{n!} (-i\varepsilon_j ar)^n  \lambda_j\cr &=e^{-i\varepsilon_jar} \lambda_j.
\end{align}
So
\begin{align}
 x=\sum_{j=1}^{N-1} (c_j e^{-i\varepsilon_j ar} \lambda_j -c_j^* e^{i\varepsilon_j ar} \lambda^\dagger_j),
\end{align}
and
\begin{widetext}
\begin{align}
[x,k_{\underline c}]=\sum_{j,k}\left( c_j c_k e^{-i\varepsilon_jar} [\lambda_j,\lambda_k]+c^*_j c^*_k e^{i\varepsilon_j ar} [\lambda^\dagger_j,\lambda^\dagger_k] - c_j c^*_k e^{-i\varepsilon_j ar} [\lambda_j,\lambda^\dagger_k]
-c^*_j c_k e^{i\varepsilon_j ar} [\lambda^\dagger_j,\lambda_k] \right) .
\end{align} 
\end{widetext}
By using (\ref{commutatore}), we see that the last two terms sum up to
\begin{align}
 -\sum_{j=1}^{N-1} |c_j|^2 J_j &i(e^{-i\varepsilon_jar}-e^{i\varepsilon_jar}) \cr &=-2  \sum_{j=1}^{N-1} |c_j|^2 \sin(\varepsilon_jar) J_j.
\end{align}
On the other hand
\begin{align}
 [\lambda_j,\lambda_k]&=[E_{j,j+1},E_{k,k+1}]\cr &=\delta_{k,j+1} E_{j,j+2}-\delta_{k+1,j} E_{j,j+2},
\end{align}
so that 
\begin{align}
 &\sum_{j,k} c_j  c_k e^{-i\varepsilon_jar} [\lambda_j,\lambda_k]\cr &=\sum_{j=1}^{N-2} c_jc_{j+1} (e^{-i\varepsilon_j ar}-e^{-i\varepsilon_{j+1}ar})E_{j,j+2},
\end{align}
and, similarly, by taking the hermitian conjugate,
\begin{align}
& \sum_{j,k} c^*_j c^*_k e^{i\varepsilon_jar} [\lambda^\dagger_j,\lambda^\dagger_k]\cr &=-\sum_{j=1}^{N-2} c^*_jc^*_{j+1} (e^{i\varepsilon_j ar}-e^{i\varepsilon_{j+1}ar})E_{j+2,j}.
\end{align}
This leads to
\begin{widetext}
 \begin{align}\label{xk}
 [x,k_{\underline c}]=&\sum_{j=1}^{N-2} \left[ c_jc_{j+1} (e^{-i\varepsilon_j ar}-e^{-i\varepsilon_{j+1}ar})E_{j,j+2}- c^*_jc^*_{j+1} (e^{i\varepsilon_j ar}-e^{i\varepsilon_{j+1}ar})E_{j+2,j}\right]\cr
 &-2  \sum_{j=1}^{N-1} |c_j|^2 \sin(\varepsilon_jar) J_j.
\end{align}
Using that the only non vanishing traces are
\begin{align}
 {\rm Tr} (E_{j,j+2} E_{k+2,k})={\rm Tr} (E_{j+2,j} E_{k,k+2})=\delta_{jk}, \quad {\rm Tr} (J_jJ_k)=-2\delta_{jk} +\delta_{j,{k+1}}+\delta_{j+1,k}.
\end{align}
we get
\begin{align}
 {\rm Tr}([x,k_{\underline c}][x,k_{\underline c}])=& -2 \sum_{j=1}^{N-2} |c_j|^2|c_{j+1}|^2 \left| e^{-i\varepsilon_j ar}-e^{-i\varepsilon_{j+1}ar} \right|^2-8\sum_{j=1}^{N-1} |c_j|^4 \sin^2 (\varepsilon_j a r)\cr
 &+\sum_{j=1}^{N-2} |c_j|^2 |c_{j+1}|^2 \sin (\varepsilon_j ar) \sin (\varepsilon_{j+1} ar).
\end{align}
\end{widetext}
Now,
\begin{align}
 \left| e^{-i\varepsilon_j ar}-e^{-i\varepsilon_{j+1}ar} \right|^2 &=2(1-\cos (ar(\varepsilon_j-\varepsilon_{j+1})))\cr &=4\sin^2 \left( ar \frac {\varepsilon_j-\varepsilon_{j+1}}2 \right).
\end{align}
Since $(\varepsilon_j-\varepsilon_{j+1})/2=0,\pm1$, we can write 
\begin{align}
 \sin^2 \left( ar \frac {\varepsilon_j-\varepsilon_{j+1}}2 \right)&=\sin^2(ar) \left(\frac {\varepsilon_j-\varepsilon_{j+1}}2 \right)^2\cr &= \frac 12 (1-\varepsilon_j \varepsilon_{j+1}) \sin^2(ar).
\end{align}
Also
\begin{align}
 \sin (\varepsilon_j ar) \sin (\varepsilon_{j+1} ar)=\sin^2 (ar) \varepsilon_j \varepsilon_{j+1}
\end{align}
and $\sin^2 (\varepsilon_j a r)=\sin^2 (a r)$, so that summing all up we get (\ref{trxkxk}). \end{proof}

\section{$SU(2)$ versus $SO(3)$.}\label{app:SO-SU}
Despite these being very well known facts, in this appendix we want to discuss the difference between $SU(2)$ and $SO(3)$, since it is crucial to identify our solutions. Locally, the two groups coincide, they have the same Lie algebra.
However, $SU(2)$ is simply connected, while $SO(3)$ is not. Indeed, $SU(2)$ is the universal covering of $SO(3)$. It has a nontrivial center $Z_{SU(2)}=\pm I$, $I$ being the unit element, and there is a surjective projection
\begin{align}
 \pi : SU(2)\longrightarrow SO(3)
\end{align}
having $Z_{SU(2)}$ as kernel. $SO(3)$ has trivial kernel and $\pi_1(SO(3))\simeq Z_{SU(2)}$. We can also write
\begin{align}
 SO(3)\simeq SU(2)/Z_{SU(2)}.
\end{align}
Now, let us illustrate the crucial difference we are interested in. Let $\tau_i$, $i=1,2,3$ a canonical basis of $Lie(G)$, $G$ being one of the two groups. We can then realise the group by means of the Euler parametrisation. This means that the generic 
element $g$ of the group has the form
\begin{align}
 g(a,b,c)= e^{a\tau_3} e^{b\tau_2} e^{c\tau_3}.
\end{align}
$a,b,c$ are the Euler angles. Each of the exponentials has a period (depending on the normalisation of the matrices), say $T_3$ for $a$ and $c$, and $T_2$ for $b$. The strategy to correctly cover $G$ exactly one time is explained in
\cite{Cacciatori:2012qi} and works as follows. To be sure to cover $G$ one integer number of times one first allow the coordinates to run each one in the respective period. This number, in general, is larger than one because of redundancies, due to
two reasons. The first reason is that the central element, parametrised by $b$, is chosen in the maximal torus (the exponential of the Cartan matrix). The redundancies correspond to the action of the Weyl group to the torus. This action is
determined by the algebra and is the same for both $SU(2)$ and $SO(3)$. It shows that indeed moving $b$ along a period quadruplicates the determination of the points for $SU(2)$ and duplicates for $SO(3)$, and one can reduce the range of 
$b$ down to $T_2/4$ or $T_2/2$ respectively. At this point,
the difference between $SU(2)$ and $SO(3)$ appears. Indeed, for $SO(3)$ this is the end of the story, it is already covered just one time, while for $SU(2)$ it remains a redundancy and we covered it twice. This redundancy is due to the fact that
\begin{align}
 e^{b\tau_2}\cap e^{c\tau_3}=
\begin{cases}
 I & {\rm if }\ G=SO(3) \\
 \Delta=e^{(T_3/2) \tau_3} & {\rm if }\ G=SU(2)
\end{cases}\ .
\end{align}
Therefore, since $\Delta^2=I$
\begin{align}
 g(a,b,c)&= e^{a\tau_3} e^{b\tau_2} e^{c\tau_3}\cr &=e^{a\tau_3} e^{b\tau_2} \Delta^{-2}e^{c\tau_3}\cr &=e^{a\tau_3} \Delta^-1e^{b\tau_2} \Delta^{-1}e^{c\tau_3}\cr &=g(a-T_3/2,b,c-T_3/2).
\end{align}
This redundancy is eliminated by reducing the range of $a$ down to $T_3/2$ for $SU(2)$. This is the way, relevant to our case, to distinguish the two kind of solutions: if the above intersection id $\Delta$, then the ranges of the variables
$a,b,c$ are $T_3/2, T_2/4, T_3$ respectively, and the group is $SU(2)$, otherwise the ranges are $T_3, T_2/2, T_3$, and the group is $SO(3)$.\\
Finally, we want to add a final remark relevant for recognising genuine solutions: for $SO(3)$ generator $\tau$ it happens of course that the orbit $\exp (x\tau)$ never meets the center, while if $\tau$ is an $SU(2)$ generator, then 
$\exp (x/2\tau)$ is the only non trivial generator of the center of $SU(2)$. No other elements of the center of $SU(N)$ can be met these kind of orbits.

\section{Representations of $SU(2)$ and periodicity}\label{periodicity}
It is well known from representation theory that spin $J$ representation of $SU(2)$ has generators $T_1$, $T_2$, $T_3$ given by the $N\times N$ matrices, with $N=2J+1$
\begin{align}
 (T_1)_{m,n}=& \frac i2 \sqrt {m(N-m)} \delta_{m,n-1} \cr +&\frac i2 \sqrt {m(N-m)} \delta_{m-1,n} \ ,  \\
 (T_2)_{m,n}=& \frac 12 \sqrt {m(N-m)} \delta_{m,n-1} \cr -&\frac 12 \sqrt {m(N-m)} \delta_{m-1,n} \ ,  \\
 (T_3)_{m,n}=& i (J+1-m) \delta_{m,n} .
\end{align}
Each of these matrices is diagonalizable with eigenvalues given by the ones of $T_3$. Since
\begin{align}
 U^\dagger \exp (x T_j) U=  \exp (x U^\dagger T_j U) 
\end{align}
it follows that the periodicity of
\begin{align}
f_j(x)= \exp (xT_j)
\end{align}
depends only on the eigenvalues and so all $f_j$ have the same periodicity, which is obviously $2\pi$ fo odd $N$ and $4\pi$ for even $N$.\\
On the other hand, let us consider the matrices $k_{\underline c}$ and $g(x)=\exp (xk_{\underline c})$. The possible periodicity of $g$ depends on the eigenvalues of $k_{\underline c}$. It is easy to see that the coefficients 
of the characteristic polynomial of $k_{\underline c}$ depend only on the $|c_j|^2$ so the phases of the $c_j$ are irrelevant for the periodicity. In particular, this means that the matrix $\exp (x\tilde T_2)$ with
\begin{align}
 (T_2)_{m,n}=&\frac {\zeta_m}2 \sqrt {m(N-m)} \delta_{m,n-1} \cr -&\frac {\zeta^*_n}2 \sqrt {n(N-n)} \delta_{m-1,n} , \quad |\zeta_j|=1,
\end{align}
has the same periodicity of $f_2(x)$.

\section{Solving the periodicity problem}\label{app:solve the problem}
In section \ref{sec:global} we have seen that for $N$ higher than 3 there is a further difficulty to overcome in order to find a global solution: generically the matrix $g(x)=e^{xk}$ is not periodic and its orbit densely fills a torus of dimension strictly higher 
than one. 
This phenomenon corresponds to the fact that the one parameter subgroup $g(x)$ is not indeed a Lie subgroup but only an imbedded subgroup. Therefore, for arbitrary choices of the coefficients $c_j$, the matrix
\begin{align}
 k_{\underline c}=\sum_{j=1}^{N-1} (c_j E_{j,j+1}-c^*_j E_{j+1,j}). 
\end{align}
cannot be used to generate global solutions unless the corresponding $g(x)$ is periodic. We will now tackle this problem in general. For the sake of completeness we will first show that no problems arise in the case $N=3$.
\subsection{The case $N=3$}\label{app:su3case}
In this simple case we have
\begin{align}
k_{\underline c}=
\begin{pmatrix}
 0 & c_1 & 0 \\
 -c_1^* & 0 & c_2 \\
 0 & -c_2^* & 0
\end{pmatrix}\ .
\end{align}
The corresponding characteristic polynomial is 
\begin{align}
 P_k(\lambda) := {\det} (\lambda \mathbb I -k_{\underline c})=\lambda(\lambda^2 +\|\underline c\|^2).
\end{align}
The eigenvalues are therefore $0,\pm i\|\underline c \|$, which are in rational ratios so $g(\gamma)=\exp(\gamma k_{\underline c})$ is periodic, in particular, with period $2\pi/\|\underline c \|$. For other purposes, we compute explicitly $g(\gamma)$. To this aim,
let us first notice that, by Cayley-Hamilton theorem, $k_{\underline c}$ satisfies
\begin{align}
 k_{\underline c} (k_{\underline c}^2 +\|\underline c\|^2 \mathbb I )=\mathbb O,
\end{align}
where $\mathbb I$ and $\mathbb O$ are the identity and the null $3\times3$ matrices. This implies $k_{\underline c}^3=-\|\underline c \|^2 k_{\underline c}$ so that any power of $k_{\underline c}$ can be reduced to a power lower than three. Hence
\begin{align}
 e^{\gamma k_{\underline c}}=g_1(\gamma) \mathbb I+g_2(\gamma) k_{\underline c}+ g_3(\gamma) k_{\underline c}^2, \label{expgammak}
\end{align}
for three functions satisfying $g_1(0)=1$, $g_2(0)=g_3(0)=0$, since $e^{\mathbb O}=\mathbb I$. Deriving (\ref{expgammak}) w.r.t. $\gamma$ and using the characteristic equation, we get
\begin{align}
g_1'(\gamma) \mathbb I&+g_2'(\gamma) k_{\underline c}+ g_3'(\gamma) k_{\underline c}^2=k_{\underline c} e^{\gamma k_{\underline c}}\cr &=g_1(\gamma) k_{\underline c}+g_2(\gamma) k_{\underline c}^2+ g_3(\gamma) k_{\underline c}^3\cr
& =(g_1(\gamma)-\|\underline c \|^2 g_3(\gamma)) k_{\underline c}+g_2(\gamma) k_{\underline c}^2,
\end{align}
so that
\begin{align}
 g_1'(\gamma)&=0, \\
 g_2'(\gamma)&= g_1(\gamma)-\|\underline c \|^2 g_3(\gamma), \\
 g_3'(\gamma)&= g_2(\gamma),
\end{align}
with the Cauchy conditions $g_1(0)=1, g_2(0)=g_3(0)=0$ (so that $g'_2(0)=1$). From the first equation we immediately get $g_1(\gamma)=1$, while deriving the second one and replacing from the third, we get
\begin{align}
 g_2''(\gamma)=-\|\underline c \|^2 g_2(\gamma), \qquad g_2(0)=0, \ g'_2(0)=1,
\end{align}
which has solution 
\begin{align}
 g_2(\gamma)= \frac {\sin (\|\underline c \|\gamma)}{\|\underline c \|}.
\end{align}
Finally, from the third equation we get
\begin{align}
 g_3(\gamma)&=\int_{0}^\gamma dx\ \frac {\sin (\|\underline c \|\gamma)}{\|\underline c \|}=\frac {1-\cos (\|\underline c \|\gamma)}{\|\underline c \|^2}\cr &=2 \frac {\sin^2 (\frac {\|\underline c \|}2\gamma )}{\|\underline c \|^2}.
\end{align}
Therefore
\begin{align}\label{expgammakappa}
 e^{\gamma k_{\underline c}} =I+\frac {\sin (\|\underline c\|\gamma)}{\|\underline c\|} k_{\underline c} +2\frac {\sin^2 (\frac {\|\underline c\|}2 \gamma)}{\|\underline c\|^2} k^2_{\underline c}\ .
\end{align}

\subsection{The general case}\label{app:general case}
One can in principle solve this problem as follows. Since $k$ is anti hermitian, it can be diagonalized in $\mathbb C$, with pure imaginary eigenvalues.
Moreover, if $\lambda$ is an eigenvalue, also $-\lambda=\lambda^*$ is. Therefore, if $S$ is the integer part of $N/2$ (so that $N=2S$ or $N=2S+1$ for $N$ even and odd respectively), generically we have $S$ distinct non vanishing eigenvalues. Let $U$ be
a unitary matrix such that 
\begin{align}
 k =U^\dagger \sigma U,
\end{align}
where $\sigma$ is the diagonal form of $k$, say
\begin{align}
 \sigma=
\begin{cases}
{\rm diag} (i\lambda_1,-i\lambda_1, \ldots, i\lambda_S, -i\lambda_S), & {\rm   N \ \ even} \\
{\rm diag} (i\lambda_1,-i\lambda_1, \ldots, i\lambda_S, -i\lambda_S,0), & {\rm   N \ \ odd}
\end{cases},
\end{align}
with $\lambda_j>0$.
Since
\begin{align}
 e^{xk}=e^{x U^\dagger\sigma U}=U^\dagger e^{x\sigma} U,
\end{align}
$ e^{xk}$ is periodic if and only if $e^{x\sigma}$ is. Now, $e^{T\sigma}$ is the identity iff and only if
\begin{align}
 e^{iT\lambda_j}=1
\end{align}
for all $j=1,\ldots,S$, that means $T\lambda_j=n_j 2\pi$, with $n_j$ a positive integer (obviously, we assume $T>0 $) for any $j=1,\ldots,S$. Therefore,
\begin{align}
 \frac {\lambda_j}{\lambda_k}=\frac {n_j}{n_k} \label{rational condition}
\end{align}
so that all pairs of eigenvalues must have rational quotients. Of course, this condition is satisfied for $N\leq 3$, and any choice of $\underline c$ is allowed. But for $N\geq 4$ we cannot choose the $c_j$ arbitrarily: only those values, such that $k$ admits eigenvalues with rational ratios are allowed. Notice that $\underline c$ remains defined up to a real multiplicative constant: if $t\in \mathbb R$, then $k_{t\underline c}=tk_{\underline c}$.\\
The eigenvalues are the solutions of the characteristic polynomial
\begin{align}
 P_N(x)=\det (x\mathbb I-k_{\underline c}),
\end{align}
of degree $N$ in $x$. 
Since $k_{\underline c}$ is antihermitian, its eigenvalues are purely imaginary and, moreover, if $\mu$ is a nonvanishing eigenvalue, then also $\mu^*=-\mu$ is an eigenvalue. So the non vanishing eigenvalues are in pairs and, if $N$ is odd, there is
at least one zero eigenvalue. Moreover, since in the factorization of the polynomial the nonvanishing eigenvalues $\mu$ must appear in the factors $(x-\mu)(x+\mu)=x^2-\mu^2$, we see that the general form of the polynomial must be 
\begin{widetext}
 \begin{align}
P_N(x)=
\begin{cases}
 (x^2)^{n}+a_{1} (x^2)^{n-1}+\ldots+a_{n}, & \mbox{for $N=2n$},\\
 x[(x^2)^{n}+a_{1} (x^2)^{n-1}+\ldots+a_{n}], & \mbox{for $N=2n+1$}.
\end{cases}
\end{align}
\end{widetext}
The coefficients $a_j$ are not the same for $N$ odd and for $N$ even, but it is convenient to keep the same name so that we can generically write the equation for the non vanishing eigenvalues as
\begin{align}
  y^{n}+a_{1} y^{n-1}+\ldots+a_{n-1}y+a_{n}=0, \qquad\ y=x^2. \label{eq-char}
\end{align}
We can be more precise:
\begin{prop}
 Using the notation $j\ll k$ for $k-j\geq 2$, we have
 \begin{align}
 a_1&=\|\underline c \|^2, \\
 a_k&=\!\!\!\sum_{j_1\ll j_2\ll\ldots\ll j_k} \!\! |c_{j_1}|^2 |c_{j_2}|^2\cdots |c_{j_k}|^2, \qquad k=2,\ldots, n. 
 \end{align}
\end{prop}
\begin{proof}
It can be easily proven by induction. We have already seen it for $N=3$. A direct computation shows that it is true also for $N=4$, since $P_4(x)=x^4+x^2(|c_1|^2+|c_2|^2+|c_3|^2)+|c_1|^2|c_3|^2$. Now, assume it to be true for $N$ and $N-1$. Let 
$k_{n}$ be the matrix $n\times n$ defined as $k_{\underline c}$ with components $c_1,\ldots,c_n$. This way, we se $k_n$ as a submatrix of $k_{n+1}$ obtained erasing the last row and column. Let 
\begin{align}
 P_n(x)= \det (x\mathbb I_{n\times n}-k_n).
\end{align}
Developing the determinant with the Laplace rule applied to the last row, we easily find
\begin{align}
 P_{N+1}(x)=x P_N(x)+|c_N|^2 P_{N-1}(x).
\end{align}
The first addendum contains all the monomials of the stated form but the terms containing $|c_N|^2$. The second addendum contains all the terms of the stated form containing $|c_{N}|^2$. The proposition is proved.
\end{proof}
So, for example,
\begin{align}
 P_4(x)=&x^4+x^2(|c_1|^2+|c_2|^2+|c_3|^2)+|c_1|^2|c_3|^2, \\
 P_5(x)=&x (x^4 +x^2 \|\underline c\|^2\cr +&(|c_4|^2|c_1|^2+|c_4|^2|c_2|^2+|c_3|^2|c_1|^2)), 
\end{align}
\begin{widetext}
 \begin{align}
 P_6(x)=&x^6+\|\underline c\|^2 x^4+x^2(|c_4|^2(|c_1|^2+|c_2|^2)+|c_1|^2|c_3|^2+|c_5|^2(|c_1|^2+|c_2|^2+|c_3|^2))+|c_1c_3c_5|^2.
\end{align}
\end{widetext}
Notice that, assuming that all $c_j$ are different from zero, we have always $a_n\neq 0$, so these are truly non zero eigenvalues. Now, condition (\ref{rational condition}) is equivalent to require that there must exist a positive real number $z$ and $n$
positive integers $m_j$, $j=1,\ldots,n$ such that the non vanishing eigenvalues of $k_{\underline c}$ must have the form $\lambda_j^{\pm}=\pm i m_j z.$
This happen if the solutions of (\ref{eq-char}) are
\begin{align}
 y_j=- z^2 m_j^2.
\end{align}
At this point, we can notice that the coefficient of the above polynomial can be written in terms of the solutions as:
\begin{align}
 a_1=&-\sum_{j=1}^N y_j, \\
 a_2=&\sum_{j_1<j_2} (-y_{j_1}) (-y_{j_2}), \\
& \qquad\qquad\ldots \\
 a_n=& \!\!\!\!\!\sum_{j_1<\cdots<j_n} (-y_{j_1})\cdots (-y_{j_n}).
\end{align}
\begin{widetext}
Comparing with the last proposition, we get the following set of equations for the $|c_j^2|=: \zeta_j$:
\begin{align}
 \sum_{j=1}^{N-1} \zeta_j=& z^2 \sum_{a=1}^n m_a^2, 
\end{align}
\begin{align}
 \sum_{j_1\ll\ldots \ll j_k\leq {N-1}} \zeta_{j_1} \cdots \zeta_{j_k}=&z^{2k}\!\!\!\! \!\!\! \sum_{a_1<\ldots<a_k\leq n} m_{a_1}^2\cdots m_{a_k}^2, \qquad k=2,\ldots,n.
\end{align} 
\end{widetext}
This is a set of $n$ equations in $N-1$ real positive variables. We will now show that it has generically an $(N-1-n)$-dimensional space of solutions in the interesting region, which is for $\zeta_j$ positive. To this end, we assume the generic situation
where all $m_a$ are different, and order them in an increasing sequence $m_1<m_2<\cdots<m_n$. We will show later that the condition on the $m_a$ cannot be weakened in order to get periodic solutions.
Then, we show that there is a simple solution on the boundary of the region of interest, which is (if $N$ is odd we assume the null eigenvalue to be the 
last one, $\lambda_{2n+1}=0$)
\begin{align}
 \zeta_{2a}=0, \qquad \zeta_{2a-1}=z^2 m^2_a, \quad a=1,\ldots,n. \label{punto particolare}
\end{align}
Next, {\it we claim that starting from this point, we can find a smooth family of solutions $\zeta_{2a-1}(\{\zeta_{2b}\})$ in a small open neighbourhood of $\zeta_{2b}=0$}. In particular, it implies that there are positive (by continuity) $\zeta_{2a-1}$'s 
parametrized by small positive $\zeta_b$'s. This is sufficient to show that there is generically a moduli space of real dimension $N-n-1$ for the solutions for the above system.
\begin{proof}[Proof of the claim] 
To prove prove the claim, let us consider te functions
\begin{align}
 F_1(\zeta_1,\ldots,\zeta_{N-1})=&\sum_{j=1}^{N-1} \zeta_j, \\
 F_k(\zeta_1,\ldots,\zeta_{N-1})=& \sum_{j_1\ll\ldots \ll j_k\leq {N-1}} \zeta_{j_1} \cdots \zeta_{j_k},\cr
  &\qquad k=2,\ldots,n,
\end{align}
and the square submatrix $M$ of its Jacobian defined by
\begin{align}
 \left. M_{a,b}=\frac {\partial F_a}{\partial {\zeta_{2b-1}}}\right|_{\zeta_j=\bar z_j},
\end{align}
where $\bar\zeta_j$ are defined by (\ref{punto particolare}). Therefore, we have
\begin{align}
 M_{1,b}=&1, \\
 M_{2,b}=&\sum_{c\neq b} z^2 m^2_c, \\
 M_{3,b}=&\sum_{\scriptstyle
\begin{array}{c}
c_1 < c_2 \\ c_j\neq b 
\end{array}
}
\!\!\!\!\!\!\!
z^4 m^2_{c_1} m^2_{c_2}, \\
\cdots & \cdots \cdots \cdots \\
 M_{k,b}=&\sum_{\scriptstyle
\begin{array}{c}
c_1 <\ldots <c_{k-1} \\ c_j\neq b 
\end{array}
}\!\!\!\!\!\!\!
z^{2k-2} m^2_{c_1}\cdots m^2_{c_{k-1}}, \\
\cdots & \cdots \cdots \cdots \\
 M_{n,b}=&\sum_{\scriptstyle
\begin{array}{c}
c_1 <\ldots <c_{n-1} \\ c_j\neq b 
\end{array}
}\!\!\!\!\!\!\!
z^{2n-2} m^2_{c_1}\cdots m^2_{c_{n-1}}. 
\end{align}
We want to compute the determinant of this matrix. It does not changes if we subtract the first column to all the other ones. In doing this, the first line becomes $\delta_{1,j}$, so that we can compute the determinant by applying the Laplace formula
to the first line. So, the determinant is equal to the determinant of the new matrix with the first row and the first column canceled out. To understand how this matrix appears, let us notice that the second row is
\begin{align}
 M_{2,b}-M_{2,1}=\sum_{c\neq b} z^2 m^2_c-\sum_{c\neq 1} z^2 m^2_c=z^2 (m^2_b-m^2_1),
\end{align}
and, more in general,
\begin{eqnarray}
 M_{k,b}-M_{k,1}&=&\!\!\!\!\!\!\!\!\!\sum_{\scriptstyle
\begin{array}{c}
c_1 <\ldots <c_{k-1} \\ c_j\neq b 
\end{array}
}\!\!\!\!\!\!\!
z^{2k-2} m^2_{c_1}\cdots m^2_{c_{k-1}} \nonumber \\
&&-\!\!\!\!\!\!\!\!\!\!\!\!\!\! \sum_{\scriptstyle
\begin{array}{c}
c_1 <\ldots <c_{k-1} \\ c_j\neq 1 
\end{array}
}\!\!\!\!\!\!\!
z^{2k-2} m^2_{c_1}\cdots m^2_{c_{k-1}} \\
&=&z^2(m_b^2-m_1^2)\!\!\!\!\!\!\!\!\!\!\!\!\!\! \sum_{\scriptstyle
\begin{array}{c}
c_1 <\ldots <c_{k-2} \\ 1\neq c_j\neq b 
\end{array}
}\!\!\!\!\!\!\!
 m^2_{c_1}\cdots m^2_{c_{k-2}}.\nonumber 
\end{eqnarray} 
Therefore, from the $b$-th column of the reduced matrix, $b=2,\ldots,n$, has a factor $z^2 (m^2_b-m^2_1)$ and since the determinant is multilinear on the columns, we get
\begin{align}
 \det(M)=\prod_{b=2}^n z^2 (m^2_b-m^2_1) \ \det (\tilde M),
\end{align}
where $\tilde M$ is a $(n-1)\times (n-1)$ matrix whose first row has all elements equal to 1 and 
\begin{align}
 \tilde M_{k,b}=&\sum_{
\begin{array}{c}
c_1 <\ldots <c_{k-1} \\ 1\neq c_j\neq b 
\end{array}
}\!\!\!\!\!\!\!
z^{2k-2} m^2_{c_1}\cdots m^2_{c_{k-1}}.
\end{align}
In other words, we see that $\tilde M$ is like $M$ but in one lower dimension and where $m_1$ has disappeared. We can then repeat inductively the same construction, finally arriving to the conclusion
\begin{align}
 \det(M)=\prod_{a<b} z^2 (m^2_b-m^2_a).
\end{align}
Since $m^2_a<m^2_b$ for $a<b$, we see that this determinant is different from zero. The proof of the claim then is an immediate consequence of the implicit function theorem.
\end{proof}
Going back to the $c_j$, we then see that in general
\begin{align}
 c_j= \xi_j \sqrt {\zeta_j (\underline m, \underline t)},
\end{align}
for arbitrary phases $\xi_j$, $j=1,\ldots, N-1$, with $\underline m\in \mathbb N_>^n$, $\underline t \in W\subset \mathbb R^{N-n-1}$. The parameters $t_j$ parametrize the above family of solutions. We can always assume that the integer $m_j$ are
coprime. Indeed, if $m$ is a common divisor of $m_j$ so that $m_j=m s_j$, then we can write $\underline m=m\underline s$ and $m$ can be reabsorbed in $z$. Having assumed this, we can now fix $z$ in such the way that $e^{xk_{\underline c}}$
has period $2\pi$. Indeed, since the non vanishing eigenvalues of $k_{\underline c}$ are $\lambda^\pm_j=\pm i z m_j$, since the $m_j$ are coprime, the common period of the associated exponential is $2\pi/z$. This fixes $z=1$. \\
Notice in particular that in this case
\begin{align}
 \|\underline c\|^2=\sum_{j=1}^n m_j^2\equiv \|\underline m \|^2.
\end{align}
In particular, the associated baryon number is
\begin{align}
 B=2\sigma m \|\underline m \|^2.
\end{align}
We have proved:
\begin{prop}
 For $N=2n$ or $N=2n+1$ and for any $n$-tuple of strictly increasing coprime positive integers $m_a$, $a=1,\ldots,n$,  
 the matrices $k_{\underline c}$ such that $e^{xk_{\underline c}}$ has period $2\pi$ is a family of dimension $2N-2+n$, where $n$ is the integer part of $N/2$. Beyond $\underline m$, this family is described by $N-1$ phases and 
 by $N-n-1$ real parameters varying in a set $W$, parametrizing the solutions of the system
 \begin{align}
 \sum_{j=1}^{N-1} \zeta_j=& \sum_{a=1}^n m_a^2, \\
 \sum_{j_1\ll\ldots \ll j_k\leq {N-1}} \zeta_{j_1} \cdots \zeta_{j_k}=&\!\!\!\! \!\!\! \sum_{a_1<\ldots<a_k\leq n} m_{a_1}^2\cdots m_{a_k}^2, \cr &\qquad k=2,\ldots,n.
\end{align}
Correspondingly, the fundamental Baryon number is $B_0=2\sigma \|\underline m \|^2$.
\end{prop}
One says that these matrices have a moduli space 
\begin{align}
 \mathcal M= \mathbb T^{N-1}\times W,
\end{align}
where $\mathbb T^{N-1}$ is the torus generated y the phases and $W\subset \mathbb R^{N-n-1}$ is the moduli space of the system. It is difficult to say something of general about the global properties of $W$. We will study in general the case 
$N=4$ where all computations are feasible explicitly.\\
{\bf Remark:} for $N=3$ we have $n=1$ and, therefore, only one integer $m$ that must be equal to 1 (to be ``coprime''). So $\underline c$ must have norm 1 and the fundamental Baryon number is $B=2\sigma$.

\subsection{The $N=4$ case}\label{sec:N=4}
Let us apply the above results to the case of $SU(4)$. We have $n=2$, so we expect the dimension of $W$ to be 1.
The eigenvalues equation for $k$ is
\begin{align}
 0=\lambda^4+\lambda^2\|\underline c \|^2+|c_1|^2|c_3|^2.
\end{align}
The four solutions are $i\lambda_+, i\lambda_-, -i\lambda_+, -i\lambda_-$, with
\begin{align}
 \lambda_\pm =\sqrt {\frac {\|\underline c\|^2}4+\frac {|c_1||c_3|}2}\ \pm \sqrt {\frac {\|\underline c\|^2}4-\frac {|c_1||c_3|}2}. \label{lambdapm}
\end{align}
Let $q\leq p$ a pair of positive coprime integer numbers. Then, we have to solve the system
\begin{align}
 \zeta_1+\zeta_2+\zeta_3&=p^2+q^2, \\
 \zeta_1\zeta_3&=p^2q^2.
\end{align}
Notice that this gives 
\begin{align}
 \lambda_+=p, \qquad \lambda_-=q.
\end{align}
Now, let us replace 
\begin{align}
\zeta_3=p^2q^2/\zeta_1  
\end{align}
in the first equation, so that
\begin{align}
 \zeta_1 +\frac {p^2q^2}{\zeta_1}-(p^2+q^2)=-\zeta_2.
\end{align}
Since we have to require $\zeta_2>0$, we see that it must be 
\begin{align}
 \zeta_1^2-(p^2+q^2)\zeta_1 +p^2q^2<0.
\end{align}
This is equivalent to say
\begin{align}
  q^2<\zeta_1<p^2.
\end{align}
So we can use $\tau=\sqrt{\zeta_1}$ as a modulus to represent $W$. The moduli space, including the boundary, is therefore
\begin{align}
 \mathcal M_4=\mathbb T^3\times [q,p].
\end{align}
For
\begin{align}
(e^{i\alpha_1}, e^{i\alpha_2}, e^{i\alpha_3}, \tau)\in \mathcal M_4, 
\end{align}
we have
 \begin{align}
 \underline c=\left(e^{i\alpha_1} \tau; \ e^{i\alpha_2} \sqrt {p^2+q^2-\tau^2 -\frac {p^2q^2}{\tau^2}};\ \frac {pq}\tau e^{i\alpha_3}\right).
\end{align}
The corresponding period is of course
\begin{align}
 T=2\pi 
\end{align}
and the fundamental Baryon number is
\begin{align}
 B_0=2\sigma (p^2+q^2).
\end{align}
Finally, we can compute the exponential.
Rewriting the characteristic polynomial as
\begin{align}
 P(x)=x^4+x^2(p^2+q^2)+p^2q^2,
\end{align}
we see that the matrix $k\equiv k_{\underline c}$ satisfies
\begin{align}
 k^4=-(p^2+q^2) k^2-p^2q^2 \mathbb I.
\end{align}
This implies that there must exist four functions $f_j(x)$, $j=0,1,2,3$ such that
\begin{align}
 e^{xk}=f_0(x) \mathbb I+f_1(x) k+f_2(x) k^2+ f_3(x) k^3,
\end{align}
with $f_0(0)=1$, $f_a(0)=0$, $a=1,2,3$.
From
\begin{align}
 \frac d{dx} e^{xk}=k e^{xk}
\end{align}
we get
\begin{align}
 f'_0(x) \mathbb I&+f'_1(x) k+f'_2(x) k^2+ f'_3(x) k^3 \cr &= f_0(x) k+f_1(x) k^2+f_2(x) k^3\cr &+ f_3(x) (-(p^2+q^2) k^2-p^2q^2 \mathbb I),
\end{align}
which gives the system of differential equations
\begin{align}
 f'_0&=-p^2q^2 f_3, \\
 f'_1&=f_0,\\
 f'_2&=f_1-(p^2+q^2)f_3, \\
 f'_3&=f_2, 
\end{align}
with the Cauchy conditions $f_j(0)=\delta_{j,0}$. Using the fourth equation in the third one we get
\begin{align}
 f''_3=f_1-(p^2+q^2) f_3, \quad f''_3(0)=0.
\end{align}
Deriving again and using the second equation:
\begin{align}
 f'''_3=f_0-(p^2+q^2) f'_3, \quad f'''_3(0)=1.
\end{align}
Deriving a last time and using the first equation, we finally get the Cauchy problem
\begin{align}
& f_3''''+(p^2+q^2) f''_3 +p^2q^2 f_3=0,\\
& f_3(0)=0, \quad f'_3(0)=0, \quad f''_3(0)=0, \quad f'''_3(0)=1.
\end{align}
This is easily solved and gives also $f_2=f'_3$, $f_1=f'_2+(p^2+q^2)f_3$, and finally $f_0=f'_1$. For $p>q$, we get
\begin{align}
 f_0(x)&=\frac {p^2}{p^2-q^2} \cos (qx)-\frac {q^2}{p^2-q^2} \cos (px),\\
 f_1(x)&=\frac {p^2}{q(p^2-q^2)} \sin (qx)-\frac {q^2}{p(p^2-q^2)} \sin (px),\\
 f_2(x)&=\frac {1}{p^2-q^2} (\cos (qx)-\cos (px)),\\
 f_3(x)&=\frac {1}{p^2-q^2} \left( \frac {\sin (qx)}q-\frac {\sin (px)}p \right).
\end{align}
In the case $p=q=1$ we have 
\begin{align}
 f_3(x)=-\frac 12 x\cos x +\frac 12 \sin x. 
\end{align}
This is sufficient to show that the case when $p=q$ must be excluded, since the solution is no more periodic.
 

\section{The Baryonic number}\label{app:integrals}
The Baryon number is defined by the integral
\begin{align}
 B=\frac 1{24\pi^2} \int \epsilon^{ijk} {\rm Tr} (R_i R_j R_k) \sqrt{g} dr\ d\phi\ d\gamma.
\end{align}
Now,
\begin{align*}
 \epsilon^{ijk} {\rm Tr} (R_i R_j R_k)&=\frac 3{L_r L_\gamma L_\varphi} \epsilon^{r\gamma\phi} {\rm Tr}(R_r[R_\gamma,R_\varphi])\cr &=-\frac {3\sigma m}{L_r L_\gamma L_\varphi}{\rm Tr} (h'[k_{\underline c},x]),
\end{align*}
where we used the explicit expressions for the $R_a$. After using (\ref{xk}), we get
\begin{align*}
 \epsilon^{ijk} {\rm Tr} (R_i &R_j R_k)\cr &=-\frac {6\sigma m}{L_r L_\gamma L_\varphi} \sum_{j=1}^{N-1} |c_j|^2 \varepsilon_j \sin (ar)\  {\rm Tr} (h'J_j),
\end{align*}
and using that
\begin{align*}
 -\varepsilon_j {\rm Tr} (h' J_j)= a,
\end{align*}
we finally get
\begin{align*}
 \epsilon^{ijk} {\rm Tr} (R_i R_j R_k)=\frac {6\sigma m}{\sqrt g} \|\underline c \|^2 a \sin (ar).
\end{align*}
Replacing in the integral and integrating we get
\begin{align}
 B=2m\sigma \|\underline c\|^2.
\end{align}
{\bf Remark:} the form
\begin{align}
\omega= \epsilon^{ijk} {\rm Tr} (R_i R_j R_k) \sqrt{g} dr\ d\phi\ d\gamma
\end{align}
is nothing but the pull back on the rectangular box of the volume form ${\rm Tr} (R\wedge R \wedge R)$ over the cycle, see for example \cite{Bertini:2005rc}.

\section{Minimal energy per Baryon}\label{app:minimal energy}
Let us minimise expression (\ref{formulone}) w.r.t. the $L_a$, $a=\varphi, r, \gamma$. Let us rewrite it in the form
\begin{widetext}
 \begin{align}
 g(L_\varphi,L_r,L_\gamma)=D L_\varphi L_r L_\gamma \left[ \frac {A^2}{L_\varphi^2} +\frac {B^2}{L_r^2}+\frac {C^2}{L_\varphi^2L_r^2} +\frac {M^2}{L_\gamma^2} \left( 1+\frac {\alpha^2}{L_\varphi^2} +\frac {\beta^2}{L_r^2} \right) \right],
\end{align}
where
\begin{align}
 D&=\frac {K\pi^3}{4\sigma m},  \qquad\ A=4\sigma,\qquad\ B=\frac {\|v_{\underline\varepsilon}\|}{\|\underline c\|}, \cr
 C&=\sigma\sqrt\lambda, \qquad\  M=2\sqrt 2\ m, \qquad\ \beta=\frac {\sqrt\lambda}{4}, \cr
 \alpha&= \sqrt\lambda \frac {\sigma}{\|\underline c\|} \left(\sum_{j=1}^{N-1} |c_j|^4+\sum_{j=1}^{N-2} |c_j|^2 |c_{j+1}|^2  \left(\frac 12-\frac 32\varepsilon_j \varepsilon_{j+1}\right)\right)^{\frac 12}.
\end{align}
\end{widetext}
Deriving w.r.t. $L_j$ and setting 
\begin{align}
 x=\frac 1{L_\varphi^2}, \qquad y=\frac 1{L_r^2}, \qquad z=\frac {M^2}{L_\gamma^2},
\end{align}
we get the equations for the stationary points:
\begin{align}
 A^2x+B^2y+C^2xy-z(1+\alpha^2x+\beta^2y)&=0, \\
 A^2x-B^2y+C^2xy-z(1-\alpha^2x+\beta^2y)&=0, \\
 -A^2x+B^2y+C^2xy-z(1+\alpha^2x-\beta^2y)&=0.
\end{align}
Solving the first equation w.r.t. $z$ and replacing in the remaining equations, we get
\begin{align}
z&=\frac {A^2x+B^2y+C^2xy}{1+\alpha^2x+\beta^2y},\label{zeta}\\
0&=B^2 y+B^2\beta^2y^2-\alpha^2 x^2 (A^2+C^2 y), \\
0&=A^2 x(1+\alpha^2x)-\beta^2 y^2 (B^2+C^2 x).
\end{align}
From the third equation we get
\begin{align}\label{ipsilon}
 y^2 =\frac {A^2 x}{\beta^2} \frac {1+\alpha^2 x}{B^2+C^2x},
\end{align}
which replaced in the second term of the second equation gives
\begin{align}
 (\alpha^2 x^2 C^2-B^2) (y+\frac {A^2 x}{B^2+C^2 x})=0.
\end{align}
Since we are looking for positive $x,y,z$, the second factor is strictly positive and the only allowed solution is $x=\frac {B}{\alpha C}$. Replacing in (\ref{ipsilon}) and then in (\ref{zeta}), we get

\begin{align}
 x=\frac {B}{\alpha C}, \qquad y=\frac {A}{\beta C}, \qquad z=\frac {AB}{\alpha \beta}.
\end{align}
Therefore,
\begin{eqnarray}
 \frac 1{L_\varphi^2}&=& \frac {\|v_{\underline\varepsilon}\|}{\lambda \sigma^2} \left(\sum_{j=1}^{N-1} |c_j|^4+\right.\\
 &&+\left.\sum_{j=1}^{N-2} |c_j|^2 |c_{j+1}|^2
  \left(\frac 12-\frac 32\varepsilon_j \varepsilon_{j+1}\right)
 \right)^{-\frac 12}, \nonumber \\
 \frac 1{L_r^2}&=& \frac {16}\lambda, \\
 \frac 1{L_\gamma^2}&=& \frac {2\|v_{\underline \varepsilon}\|}{\lambda m^2} \left(\sum_{j=1}^{N-1} |c_j|^4+\right.\\
 &&+\left. \sum_{j=1}^{N-2} |c_j|^2 |c_{j+1}|^2
  \left(\frac 12-\frac 32\varepsilon_j \varepsilon_{j+1}\right)
 \right)^{-\frac 12},\nonumber
\end{eqnarray}
\begin{widetext}
and the corresponding energy per Baryon in standard units ($K=(6\pi^2)^{-1}$, $\lambda=1$) is
\begin{align}\label{energia minima}
 g(\underline c,\varepsilon)=\frac \pi{3\sqrt 2} \left[ 2+\frac {\|v_{\underline\varepsilon}\|}{\|\underline c\|^2}  \left(\sum_{j=1}^{N-1} |c_j|^4+\sum_{j=1}^{N-2} |c_j|^2 |c_{j+1}|^2
  \left(\frac 12-\frac 32\varepsilon_j \varepsilon_{j+1}\right)
 \right)^{\frac 12} \right].
\end{align} 
\end{widetext}

\bibliography{skyrmions}

\end{document}